\documentclass[a4paper]{article}

\usepackage[utf8]{inputenc}
\usepackage[T1]{fontenc}
\usepackage{float}
\usepackage{amsmath}
\usepackage{amssymb}
\usepackage{latexsym}
\usepackage{cmmib57}
\usepackage{amscd}
\usepackage{graphicx}
\usepackage{pstricks}
\usepackage{pst-coil}
\usepackage{pst-node}
\usepackage{float}
\usepackage{wrapfig}
\usepackage{epsf}

\newgray{meascolor}{.5}
\newgray{opercolor}{.5}
\newgray{entacolor}{.5}
\newgray{qbitcolor}{.1}
\newgray{mklight}{.7}
\newgray{mkfill}{.9}
\psset{unit=\unitlength}
\psset{arrowsize=2pt 5}
\psset{coilarm=.4}

\newcommand{\scr}[1]{\mathcal{#1}}
\def\idty{{\leavevmode{\rm 1\ifmmode\mkern -4.8mu\else\kern -.3em\fi
      I}}}
\renewcommand{\Bbb}[1]{\if1#1\idty\else\mathbb{#1}\fi}

\newcommand{\kb}[1]{|#1\rangle\langle#1|}
\newcommand{\KB}[2]{|#1\rangle\langle#2|}

\newcommand{\tr}{\operatorname{tr}}
\newcommand{\supp}{\operatorname{supp}}

\newcommand{\Lz}{\operatorname{L}^2}

\newcommand{\SP}{\operatorname{span}}
\newcommand{\diag}{\operatorname{diag}}
\newcommand{\Ran}{\operatorname{Ran}}

\newcommand{\Pf}{\operatorname{Pf}}

\newcommand{\sign}{\operatorname{sign}}

\newtheorem{thm}{Theorem}[section]

\newtheorem{prop}[thm]{Proposition}
\newtheorem{lem}[thm]{Lemma}
\newtheorem{kor}[thm]{Corollary}

\newenvironment{proof}{\par\noindent\textit{Proof.\ }}{\hfill $\Box$ \vspace{1em}}

\newtheorem{aX}{Axiom} %\renewcommand{\theaX}{E\Roman{aX}}

\newcommand{\J}{J}

\title{Entanglement distillation from quasifree Fermions.}
\author{Zolt\'an K\'ad\'ar\footnote{kadar@isi.it}, Michael Keyl\footnote{keyl@isi.it}\, and Dirk
  Schlingemann\footnote{dirk@isi.it}\\ 
{\small ISI Foundation, Viale S. Severo 65 - 10133 Torino - Italy}}

\begin{document}
\maketitle

\begin{abstract}
  We develop a scheme to distill entanglement from bipartite Fermionic
  systems in an arbitrary quasifree state. It can be applied if either one
  system containing infinite one-copy entanglement is available or if an
  arbitrary amount of equally prepared systems can be used. We show that the
  efficiency of the proposed scheme is in general very good and in some cases
  even optimal. Furthermore we apply it to Fermions hopping on an infinite
  lattice and demonstrate in this context that an efficient numerical analysis
  is possible for more than $10^6$ lattice sites.
\end{abstract}

%%%%%%%%%%%%%%%%%%%%%%%%%%%%%%%%%%%%%%%%%%%%%%%%%%%%%%%%%%%%%%%%%%%%%%%%%%%%%%%%
\section{Introduction}
\label{sec:introduction}
%%%%%%%%%%%%%%%%%%%%%%%%%%%%%%%%%%%%%%%%%%%%%%%%%%%%%%%%%%%%%%%%%%%%%%%%%%%%%%%%
An important class of Fermionic states are the Gaussian ones. They describe
equilibrium states of quasifree spin chains, such as non-interacting electrons
\cite{LSM}. 
Since they are completely characterized by their two-point correlations,
they are under good control also in case of large particle number. For an example
of their standard description based on a fixed basis in the Fermionic Fock space, see
\cite{b}.

Defining entanglement is not straightforward in Fermionic systems as Fermions
are indistinguishable and described by antisymmetric tensor products. There
are several ways around this problem. E.g. in \cite{sckll} antisymmetrized
states with definite particle number are studied while in \cite{AFOV} the focus
is on the lattice spin system admitting Fermionic description. The
conceptionally clearest approach is, however, to base the description of
subsystems not on tensor products of Hilbert spaces, but on the notion of
local observables \cite{HAAG92}. This approach is successfully applied
 to the study of entanglement of systems with infinite degrees of freedom
\cite{MR2153773,KMSW06} and for the analysis of separable \cite{Mo06}
 and maximally entangled \cite{MEF08} states of Fermionic systems. 
Furthermore, the conservation of local parity superselection rule allows for
different possible definitions for separability \cite{BaCiWo07} (consult that
article also for more literature on Fermionic entanglement). 

In this paper we will take an operational point of view. Instead of asking how
much entanglement is contained in a given bipartite, Fermionic system, we will
look for explicit distillation protocols, i.e. procedures to generate (in
terms of LOCC), from a number of Fermions, pairs of \emph{distinguishable}
particles in a maximally entangled state. The asymptotic rate of such a
protocol can then be regarded as a measure for the entanglement contained in
the original state. The advantage of this approach is that the only concept
which needs a (slight) generalization is that of local operations (of
bipartite Fermionic systems and from a Fermionic to an ordinary bipartite
system). The latter, however, can be easily based on the idea of local
observables mentioned in the last paragraph; cf. \cite{K09} for a more
complete discussion.

Following this idea, we present a family of  protocols which is
particularly adopted towards distillation from quasifree (i.e. Gaussian)
states. The general structure involves a two step procedure: First we
generate (with a certain success probability $p$) bipartite $d$-level systems 
(where the local dimension can be chosen within a certain range) in an
isotropic state $\sigma$. In particular for very large systems (cf. the
discussion of free Fermions on an infinite lattice in Section
\ref{sec:free-ferm-hopp-1}) the entanglement fidelity of $\sigma$ can be
already very close to one such that no more steps are necessary. If this is
not the case but $\sigma$ is distillable we can continue with standard
techniques like the recurrence method or hashing (cf. \cite{VoWo03} and the
references therein).

Each protocol in this family admits a very easy parametrization in terms
of two operators (a projection $D$ and a partial isometry $V$; cf. Section
\ref{sec:family-dist-prot}) on the one-particle reference space $\mathcal{K}$
rather than the corresponding Fermionic Fockspace $\mathcal{H}$. This implies
in particular that we can express the success probability $p$ and the fidelity
$f$ of the isotropic output states $\sigma$ (from which the overall rate of
the protocol can be calculated with known formulas, once we have decided how to
process the resulting $d$-level systems) in terms of $D$ and $V$. The
dimension of $\mathcal{K}$ (i.e. the space on which $D$ and $V$ operate)
grows, however, only linearly in the system size (i.e. the number of
independent modes) while $\dim \mathcal{H}$ grows exponentially. Hence we get
a very efficient way to discuss the entanglement content of quasifree
Fermionic systems which is applicable to, literally, millions of
particles. This is explicitly demonstrated in Section
\ref{sec:free-ferm-hopp-1}. 

The given class of protocols is still large and for a given quasifree state
$\rho_S$ (which is completely described by its two-point correlation matrix
$S$) most of them lead to very poor results. Therefore, in Section 
\ref{sec:optimal-protocol} we present 
 a scheme to derive the operators $D$ and $V$ from
$S$. It is based on the observation that to any generic, quasifree state
$\rho_S$  we can associate a quasifree, maximally entangled state in a natural
way (this resembles a little bit the Schmidt decomposition, but note that
pureness of $\rho_S$ is not required). For a large subset of quasifree states
(whose dimension grows quadratically in the system size) this particular
protocol optimizes the product of success probability $p$ and fidelity $f$ (we
will discuss in Section \ref{sec:optimal-protocol} why this is a good figure
of merit and, actually, related to the overall distillation rate). 

The presented scheme is discussed in terms of two classes of examples: Small
systems consisting of two or four modes and free Fermions, hopping on a
one-dimensional lattice. The numerical efficiency mentioned above allows us in
the latter case to treat large particle numbers and this leads to new insights
about the entanglement properties of this particular system, which are
interesting in their own right.

The paper is organized as follows: In Section \ref{sec:math-prel} we will
start with some mathematical preliminaries about Fermionic systems. The
approach we present here is based on the ``selfdual formalism''
\cite{Araki87,MR0295702} and might be unfamiliar for most readers (although it is
not particularly new). For the purposes of this paper it is, however, the
most natural formalism. In Section \ref{sec:bipart-syst-entangl} we apply it
to discuss some aspects of bipartite Fermionic systems and in Section
\ref{sec:family-dist-prot} we present the family of protocols mentioned
above. The question of optimality is then treated in Section
\ref{sec:optimal-protocol}. The rest of the paper is finally devoted to
examples: two and four modes in Section \ref{sec:two-four-modes} and 1d
lattice systems in Sections \ref{sec:dist-from-latt-1} and
\ref{sec:free-ferm-hopp-1}. Several technical aspects (including, in particular,
proofs) are postponed to the appendix.

%%%%%%%%%%%%%%%%%%%%%%%%%%%%%%%%%%%%%%%%%%%%%%%%%%%%%%%%%%%%%%%%%%%%%%%%%%%%%%%%
\section{Mathematical preliminaries}
\label{sec:math-prel}
%%%%%%%%%%%%%%%%%%%%%%%%%%%%%%%%%%%%%%%%%%%%%%%%%%%%%%%%%%%%%%%%%%%%%%%%%%%%%%%%

To introduce some terminology and notations, let us start with some technical
remarks about Fermionic systems. All the material we are going to present here
can be found in the literature, \cite{Araki87,MR0295702} are some standard
references.  

A Fermionic system consisting of $n$ modes can be described by smeared
Majorana fields
\begin{equation}
  \mathcal{K}^{(n)} \ni x \mapsto B(x) \in \mathcal{B}(\mathcal{H}^{(n)}).
\end{equation}
The system Hilbert space $\mathcal{H}^{(n)}$ is $2^n$ dimensional and can be
realized as the Fermionic Fockspace over $\Bbb{C}^n$. Therefore we will refer
to it occasionally as the \emph{Fermionic Fock space of $n$ modes} although 
the precise form of $\mathcal{H}^{(n)}$ is (apart from its dimension) not
really important. The Hilbert space $\mathcal{K}^{(n)}$ -- in the following
called \emph{reference space} -- is $2n$ dimensional and equipped with an
antilinear involution (called \emph{complex conjugation}) $J:
\mathcal{K}^{(n)} \rightarrow \mathcal{K}^{(n)}$. In other words $J$ satisfies
\begin{equation}
  J(x + \lambda y) = J(x) + \bar{\lambda} J(y),\quad J^2 = \Bbb{1}
\end{equation}
for all $x,y \in \mathcal{K}^{(n)}$ and all $\lambda \in \Bbb{C}$. The typical
choice for $\mathcal{K}^{(n)}$ and $J$ is $\Bbb{C}^{2n}$ equipped with the
ordinary conjugation in the canonical basis. As with $\mathcal{H}^{(n)}$,
however, the explicit realization of $(\mathcal{K}^{(n)},J)$ is not
important. $\mathcal{K}^{(n)}$ contains a distinguished \emph{real} subspace
given by
\begin{equation}
  \mathcal{K}^{(n)}_{\Bbb{R}} = \{ x \in \mathcal{K}^{(n)} \, | \, Jx = x \}
  \subset \mathcal{K}^{(n)}.
\end{equation}
Lots of structures we will encounter in the following are actually associated
to $\mathcal{K}_{\Bbb{R}}^{(n)}$ rather than to $\mathcal{K}^{(n)}$. We will
call in particular an orthonormal basis $e_j$, $j=1,\dots,2n$ of
$\mathcal{K}_{\Bbb{R}}^{(n)}$ (which is of course an orthonormal basis of
$\mathcal{K}^{(n)}$, too) a \emph{real basis.}

The operators $B(x)$ are (complex) linear in $x \in \mathcal{K}^{(n)}$ satisfy
the \emph{canonical anticommutation relations} (CAR) in the form
\begin{equation} \label{eq:76}
  \{B(x),B(y)\} = \langle Jx, y \rangle \Bbb{1},\quad B(x)^* = B(Jx)
\end{equation}
and they act \emph{irreducibly} on $\mathcal{H}^{(n)}$, i.e.
\begin{equation} \label{eq:77}
  [A,B(x)] =  0 \ \forall x \in \mathcal{K}^{(n)} \Rightarrow A = \lambda \Bbb{1}.
\end{equation}
These conditions fix the $B(x)$ up to unitary equivalence, which is the reason
why we are not interested in their explicit form. They can be constructed
easily in terms of ordinary creation and annihilation operators. The details
are shown in Appendix \ref{sec:self-dual-formalism}.

Let us consider now states of the system. To any density operator $\rho$ on
$\mathcal{H}^{(n)}$ we can associate a \emph{covariance operator} $S \in
\mathcal{B}(\mathcal{K}^{(n)})$ by  
\begin{equation} 
  \tr(\rho B(x) B(y) ) = \langle J x, S y\rangle\quad \forall x,y \in
  \mathcal{K}. 
\end{equation}
$S$ is selfadjoint and satisfies
\begin{equation} \label{eq:45}
  0 \leq S \leq \Bbb{1}\ \text{and}\ JSJ = \Bbb{1} - S.
\end{equation}
A state $\rho_S$ is called \emph{quasifree} if it is uniquely characterized
by $S$ and the conditions 
\begin{gather} \label{eq:38}
  \tr(\rho_S B(x_1) \cdots B(x_{2l+1})) = 0\\
  \tr(\rho_S B(x_1) \cdots B(x_{2l})) = \sum  \sign( p ) \prod_{j=1}^{l} \langle
  \J x_{p(2j-1)} , S x_{p(2j)} \rangle, 
\end{gather}
which have to hold for all $l \in \Bbb{N}$ and $x_k \in \scr{K}^{(n)}$. The sum in
(\ref{eq:38}) is taken over all permutations $p$ satisfying  
\begin{equation}
  p(1) < p(3) < ... < p(2l-1) , \quad p(2j-1) < p(2j) 
\end{equation}
and $\sign (p)$ is the signature of $p$. 

A pure state $\psi_E \in \mathcal{H}^{(n)}$ is quasifree (and then called a
\emph{Fock state}) iff its covariance operator $E$ is a projection. According
to Equation (\ref{eq:45}) it has to satisfy $JEJ = \Bbb{1} - E$. Each
projection with this property is called a \emph{basis projection}. 

Consider now a unitary operator $R$ on $\mathcal{K}$ satisfying $[J,R] = 0$. It
leaves the real subspace $\mathcal{K}_{\Bbb{R}}$ invariant; i.e. it is a
\emph{real orthogonal} transformation. $R$ gives rise to a new set of
operators by $B_R(x) = B(Rx)$. It is easy to see that $B_R$ is complex linear
and satisfies (\ref{eq:76}) and (\ref{eq:77}). Hence the fields $B(x)$ and
$B_R(x)$ are unitarily equivalent and describe effectively the same physical
system. More precisely there is a unitary $\Gamma(R)$ on $\mathcal{H}^{(n)}$,
called the \emph{Bogolubov transformation} of $R$ satisfying 
\begin{equation} \label{eq:46}
  \alpha_R(B(x)) = \Gamma(R) B(x) \Gamma(R)^* = B(Rx).
\end{equation}
The \emph{Bogolubov automorphism} $\alpha_R$ of
$\mathcal{B}(\mathcal{H}^{(n)})$ is uniquely determined by this condition, the
Bogolubov transformation is only fixed up to a phase.

The most important special case arises with $R=-\Bbb{1}$. The corresponding
automorphism $\Theta = \alpha_{-\Bbb{1}}$ is called the \emph{parity
  automorphism}. According to (\ref{eq:46}) it is characterized by
\begin{equation} \label{eq:47}
  \Theta(B(x)) = \theta B(x) \theta  = - B(x).
\end{equation}
The associated Bogolubov transformation $\Gamma(-\Bbb{1}) = \theta$
can be chosen selfadjoint and is then called \emph{parity operator}. It is
fixed by (\ref{eq:47}) up to a sign and given in terms of a real basis $e_a$, 
$a=1,\dots,2n$ of $\mathcal{K}^{(n)}$ by (cf. Appendix
\ref{sec:parity-operator})
\begin{equation} \label{eq:50}
  \theta = 2^n i^n B(e_1) \cdots B(e_{2n}).
\end{equation}
If we change the basis in terms of an orientation preserving, real, orthogonal
transformation, $\theta$ remains invariant. If we change the orientation, the
operator $\theta$ changes the sign. In the following we will assume that a
particular orientation is chosen.

The parity automorphism gives rise to the distinction of even and odd
operators: $A \in \mathcal{B}(\mathcal{H}^{(n)})$ is called even if $\Theta(A)
= A$ and odd if $\Theta(A) = -A$. Only selfadjoint, even operators can be
regarded as observables. Similarly, a completely positive (cp) map
$T:\mathcal{B}(\mathcal{H}^{(n)}) \rightarrow \mathcal{B}(\mathcal{H}^{(n)})$ is
an operation only if it commutes with $\Theta$, and  only its action on even
elements is relevant. In other words, if a second cp map $T_1$ satisfies $T(A) =
T_1(A)$ for all $A$ with $\Theta(A)=A$ it describes the same operation.

Using the \emph{spectral decomposition} $\theta = P_+ - P_-$ of the parity
operator we can decompose the Hilbert space $\mathcal{H}^{(n)}$ into an even
and an odd part:
\begin{equation} \label{eq:49}
  \mathcal{H} = \mathcal{H}^{(n)}_+ \oplus \mathcal{H}_-^{(n)},\quad
  \mathcal{H}^{(n)}_\pm = P_\pm \mathcal{H}^{(n)},\quad \theta = P_+ - P_-.
\end{equation}
An operator $A = \mathcal{B}(\mathcal{H}^{(n)})$ is even iff it is of the form
$A = A_+ \oplus A_-$ with $A_\pm \in \mathcal{B}(\mathcal{H}^{(n)}_\pm)$.

A \emph{Fermionic subsystem} of our given system consists of a number of
modes which are distinguished by a certain physical property like position in
space. Mathematically it is described by a projection $D$ satisfying 
\begin{equation}
  D \in \mathcal{B}(\mathcal{K}),\quad D^2=D,\quad D^*=D\quad [D,J] = 0.
\end{equation}
The last condition implies that $D$ can be restricted to the real subspace
$\mathcal{K}^{(n)}_\Bbb{R}$ of $\mathcal{K}^{(n)}$. Therefore we will call it a
\emph{real projection}. $D$ projects onto the subspace $D\mathcal{K}$ of
$\mathcal{K}^{(n)}$ containing the modes belonging to the subsystem. The
corresponding Majorana operators does not act any longer
irreducibly on $\mathcal{H}^{(n)}$. Instead, we can find a unitary 
\begin{equation}
  U_D: \mathcal{H}^{(n)} \rightarrow \mathcal{H}^{(l)} \otimes
  \mathcal{H}^{(n-l)},\quad l = \tr(D) < n
\end{equation}
which satisfies and is (up to a phase) uniquely characterized by
\begin{equation} \label{eq:48}
  U_D B(x) U_D^* =
  \begin{cases}
    B(x) \otimes \Bbb{1} &  x \in D \mathcal{K}^{(n)} \cong \mathcal{K}^{(l)} \\
    \theta \otimes B(x) & x \in (\Bbb{1} - D) \mathcal{K}^{(n)} \cong \mathcal{K}^{(n-l)}.
  \end{cases}
\end{equation}
Here $\mathcal{H}^{(l)}$ and $\mathcal{H}^{(n-l)}$ denotes the Fockspace of $l
= \tr(D)$ and $n-l = \tr(\Bbb{1} - D)$ Fermionic modes. Similarly
$\mathcal{K}^{(l)}$ and $\mathcal{K}^{(n-l)}$ are the corresponding reference
spaces. 

Dropping the complementary subsystem 
is described by the operation (in the Heisenberg picture)
\begin{equation} \label{eq:51}
  \mathcal{B}(\mathcal{H}^{(l)}) \ni A \mapsto \Delta_D(A) = U_D^* A \otimes
  \Bbb{1} U_D \in \mathcal{B}(\mathcal{H}^{(n)}). 
\end{equation}
It is easy to see that even operators on $\mathcal
{H}^{(l)}$ are mapped to
even operators on $\mathcal{H}^{(n)}$. Hence the map describes an operation in
the sense discussed above. 

\section{Bipartite systems and entanglement}
\label{sec:bipart-syst-entangl}
%%%%%%%%%%%%%%%%%%%%%%%%%%%%%%%%%%%%%%%%%%%%%%%%%%%%%%%%%%%%%%%%%%%%%%%%%%%%%%%%

Let us consider now a bipartite systems consisting of $2m$ modes shared by
Alice and Bob. Hence $n = 2m$ and there is a distinguished $n$-dimensional, real  
projection $I_A \in \mathcal{B}(\mathcal{K}^{(2m)})$  which defines the Alice
subsystems. Likewise $I_B = \Bbb{1} - A$ defines the Bob subsystem, and the
Hilbert space $\mathcal{K}^{(2m)}$ decomposes into a direct sum
\begin{equation}
  \mathcal{K}^{(2m)} = \mathcal{K}^{(m)}_A \oplus \mathcal{K}^{(m)}_B,\ J =
  J_A \oplus J_B,\ \mathcal{K}_{A/B}^{(m)} = I_{A/B} \mathcal{K}^{(2m)},\ J_{A/B} = I_{A/B} J
  I_{A/B}. 
\end{equation}
According to Equation (\ref{eq:48}) we can identify\footnote{We will do this
  in the following without further reference, as long as confusion can be
  avoided.} the Hilbert space $\mathcal{H}^{(2m)}$ via the unitary $U_A =
U_{I_A}$ with $\mathcal{H}^{(m)} \otimes \mathcal{H}^{(m)}$ such that ordinary
entanglement theory applies. There are, however, two important caveats. 

According to (\ref{eq:48}) a Majorana operator $B(x)$ belonging to Bob's
subsystem, i.e. $x \in I_B \mathcal{K}^{(2m)}$ is of the form $U_A B(x) U_A^* =
\theta \otimes B(x)$. Hence Alice's Hilbert space is \emph{not}
invariant under the action of Bob's operators\footnote{Semingly, a similar
  problem does not arise with Alice's operators. This is, however, only an
  artifact of the special representation given by the unitary $U_A$. We could
  equally well choose $U_B = U_{I_B}$ to exchange the roles of Alice and
  Bob. The crucial point is that Majorana operators  $B(x_{A/B})$ with
  $x_{A/B} \in \mathcal{K}^{(m)}_{A/B}$ do not act independently.}. This
problem can be rectified if we take the remark seriously that only \emph{even}
operators are physically relevant. Since an operator $Q \in
\mathcal{B}(\mathcal{H}^{(2m)})$ is even iff it can be written as an even
polynomial in the $B(x)$, it follows immediately that an even operator
belonging to the Alice (Bob) subsystem operates trivially on the Bob (Alice)
Hilbert space. Or in other words the \emph{observable algebras} 
\begin{gather} \label{eq:72}
  \mathcal{A} = \SP \{B(x_1) \cdots B(x_{2l}) \, | \, x_1, ..., x_{2l} \in I_A
  \mathcal{K},\ l \in \Bbb{N} \} \\
  \mathcal{B} = \SP \{B(x_1) \cdots B(x_{2l}) \, | \, x_1, ..., x_{2l} \in I_B
  \mathcal{K},\ l \in \Bbb{N} \} \label{eq:74}
\end{gather}
associated to Alice and Bob respectively are of the form
\begin{gather}
  U_A \mathcal{A} U_A^* = \left(\mathcal{B}(\mathcal{H}^{(m)}_+) \oplus
  \mathcal{B}(\mathcal{H}^{(m)}_-)\right) \otimes \Bbb{1} \subset
  \mathcal{B}(\mathcal{H}^{(m)}) \otimes \Bbb{1} \\
  U_A \mathcal{B} U_A^* = \Bbb{1} \otimes
  \left(\mathcal{B}(\mathcal{H}^{(m)}_+) \oplus
    \mathcal{B}(\mathcal{H}^{(m)}_-)\right) \subset
  \Bbb{1} \otimes \mathcal{B}(\mathcal{H}^{(m)}).
\end{gather}
But now a second problem arises, since both parties admit observables -- the
local parities $\theta \otimes \Bbb{1}$ and $\Bbb{1} \otimes
\theta$ -- which can be measured \emph{without disturbing the system.}
Therefore we have to deal with entanglement theory at the presence of local
superselection rules. A  detailed discussion of this subject can be found in
\cite{SVC}. For us only a few special topics are relevant.

Let us start with a quasifree state $\rho_S$ with covariance operator $S$. The
subsystem projections $I_A$ and $I_B$ give rise to the operators $S_{jk} = I_j
S I_k$, $j,k=A,B$. The properties of $S$ imply that
\begin{equation} \label{eq:55}
  X_S = -i(2S_{AA} - \Bbb{1}),\quad Z_S = -i (2S_{BB} - \Bbb{1}),\quad Y_S = -2i S_{AB}
  = 2i S_{BA}^* 
\end{equation}
are real operators and that $X_S$ and $Z_S$ are antisymmetric. 

Many entanglement properties of $\rho_S$ can be stated in terms of $X_S, Y_S$ and
$Z_S$. Important for us are maximally entangled, quasifree states which are
according to \cite{MEF08} characterized by the condition
\begin{equation}
  X_S = Z_S = 0,\quad Y_SY_S^* = I_A,\quad Y_S^*Y_S = I_B.
\end{equation}
In other words $Y_S$ is a partial isometry with $I_A$ as source and $I_B$ as its
target projection. If we introduce a real basis $e_a$, $a=1,\dots,4m$ of
$\mathcal{K}^{(2m)}$ which is adopted to the Alice/Bob split, i.e.
\begin{equation}
  e_1, \dots, e_{2m} \in I_A \mathcal{K}^{(2m)},\quad e_{2m+1}, \dots, e_{4m} \in I_B
  \mathcal{K}^{(2m)} 
\end{equation}
we get the matrix
\begin{equation}
Y_{jk} = \langle e_j, Y_S e_{2m + k} \rangle, \quad j,k = 1, \dots, 2m\label{bpp}
\end{equation}
which is orthogonal. This gives us a parametrization of the set of all
quasifree, maximally entangled states in terms of the group $\mathrm{O}(2m)$,
and it implies that there are two possible orientations. This is connected to
the \emph{global} sign of $\theta$ which can be fixed by the condition
$\tr(\theta \rho_S)  > 0$ (provided this expectation value is not zero; note
in this context that $\tr(\rho_S \theta) = \pm 1$ only, occurs iff $\rho_S$
is pure, i.e. a Fock state).  

According to the definition each quasifree state is an \emph{even} state,
i.e. $\tr(\rho_S A) = 0$ for each odd operator $A$ on
$\mathcal{H}^{(2m)}$. This implies $P_+\rho_SP_- =P_-\rho_SP_+ = 0$, and $P_+
\psi_E = 0$ or $P_- \psi_E = 0$ if $\rho=\kb{\psi_E}$ is pure. Note that the
roles of $P^+$ and $P^-$ can be exchanged by switching the orientation (and
therefore the sign of $\theta$; cf. the discussion of Equation (\ref{eq:50})).
If we fix the sign of $\theta$, as stated in the last paragraph, by the
condition that $\langle \psi_E, \theta \psi_E \rangle = 1$ holds we get $P^-
\psi_E = 0$, which should hold in the following.

Now consider the decomposition $\theta^{(2m)} = \theta^{(m)} \otimes
\theta^{(m)}$ of the global parity into a product of local ones. It implies a
likewise decomposition of the projections $P_\pm$ as 
\begin{equation} \label{eq:52}
P_+ = P_{++} + P_{--},\quad P_-=P_{+-} + P_{-+},\quad P_{jk} = P_j \otimes
P_k,\ j,k=\pm
\end{equation}
and therefore $P_- \psi_E = 0$ implies
\begin{equation} \label{eq:54}
\psi_E = \frac{1}{\sqrt{2}} \left( \psi_{E,+} + \psi_{E,-} \right),\quad
\psi_{E,k} \in \mathcal{H}_{k}^{(m)} \otimes \mathcal{H}_k^{(m)},\ k=\pm
\end{equation}
and $\psi_E$ is maximally entangled if $\psi_{E,\pm}$ are maximally entangled
in the usual sense\footnote{Note, however, that not all maximally entangled
states can arise here, since $\psi_E$ is quasifree by assumption.}. 

An important point is now that a relative phase between $\psi_{E,+}$ and
$\psi_{E,-}$ can not be determined by \emph{local} measurements by Alice and
Bob. In other words $\psi_E$ and
\begin{equation}
\tilde{\psi} = \frac{1}{\sqrt{2}} \left( \psi_{E,+} + e^{-i\alpha}
\psi_{E,-} \right) 
\end{equation}
are completely equivalent, as long as only LOCC operations and measurement
are allowed. The only quasifree choice in this family arises for $\alpha =
\pi$. The corresponding basis projection $\tilde{E}$ differs from $E$ only
by the sign of the off-diagonal block $\tilde{E}_{AB}$, i.e.
\begin{equation}
\tilde{S}_{AA} = S_{AA},\quad \tilde{S}_{BB} = S_{BB},\quad \tilde{S}_{AB} =
- S_{AB}. 
\end{equation}
If $\psi_E$ is maximally entangled $\psi_{\tilde{E}}$ is maximally entangled
as well and it represents the same orientation, i.e. $\langle
\psi_{\tilde{E}}, \theta \psi_{\tilde{E}}\rangle = \langle\psi_E, \theta
\psi_E\rangle$. 

Finally, let us have a short look at LOCC. As with
entanglement, the usual concepts apply to Fermions if we use the tensor
product decomposition of $\mathcal{H}^{(2m)}$ given by $U_A$ and take into
account that only operators in the tensor product $\mathcal{A} \otimes
\mathcal{B}$ are physically relevant. We are not giving a full discussion
here, but refer the reader to \cite{K09}. Instead we will only have a short
look on those operations which are relevant for our purposes.

Our first example is ``dropping a subsystem'' as described in Equation
(\ref{eq:51}). If the subsystem projection $D$ commutes with $I_A$ we can
decompose $D$ as $D = D_A \oplus D_B$ with $D_A = I_AD$ and $D_B = I_BD$. The
overall operation $\Delta_D$ can therefore be written as $\Delta_{D_A} \otimes
\Delta_{D_B}$. The final system consists of $\tr(D_A) + \tr(D_B)$ Fermionic
modes. 

Now consider the projection operators $P_{jk}$, $j,k=\pm$ introduced in
Equation (\ref{eq:52}). They define a von Neumann-Lüders instrument which we
will call in the following a \emph{joint parity measurement}. The possible
values are $++,+-,-+,--$, and if the system is before the measurement in 
the state $\rho$ we get
\begin{equation}
p^{jk} = \tr(P^{jk} \rho)\ \text{and}\ \rho^{jk} = \frac{P^{jk} \rho P^{jk}}{p^{jk}}
\end{equation}
for the probability $p^{jk}$ to get the outcome $jk$ and for the corresponding
posteriori state $\rho^{jk}$. Each of the subchannels $\rho \mapsto
P_{jk}\rho P_{jk}$ is obviously a local operation. They produce bipartite
systems, which are described by the Hilbert spaces
\begin{equation}
\mathcal{H}_j^{(m)} \otimes \mathcal{H}_k^{(m)} \cong \Bbb{C}^d\otimes \Bbb{C}^d,\quad
d = 2^{m-1}.
\end{equation}
Hence we get a pair of (distinguishable) $d-$level systems. For later
reference let us introduce as  well the probability 
\begin{equation} \label{eq:53}
p = p_{++} + p_{--} = \tr\bigl(\rho (P_{++} + P_{--})\bigr) = \tr(\rho P_+)
\end{equation}
to get the same parity on both sides. Since $P_{++} + P_{--} = P_+$ this
definition depends on the choice of an orientation (i.e. the sign of
$\theta$). 

An LOCC operation which produces again a Fermionic system (but in contrast to
$\Delta_D$ of the same type) is \emph{twirling}. We choose a quasifree,
maximally entangled state $\rho_E = \kb{\psi_E}$ and average over the group
$G$ of all local unitaries $U=U_A\otimes U_B$ on $\mathcal{H}^{(2m)} =
\mathcal{H}^{(m)} \otimes \mathcal{H}^{(m)}$ such that $U_A \mathcal{A} U_A^*
= \mathcal{A}$, $U_B \mathcal{B} U_B^* = \mathcal{B}$, and $U \psi_{E} =
\psi_{E}$ holds. This leads to a state  
\begin{align}
T_E \rho &= \int_G U \rho U^* dU \\
&= \lambda_+ |\psi_E\rangle\langle\psi_E| + \lambda_-
|\psi_{\tilde{E}}\rangle\langle\psi_{\tilde{E}} | +   \mu_+ (P^{++} + P^{--}) + \mu_-(P^{+-} +
P^{-+}) \label{eq:7}
\end{align}
with
\begin{gather} \label{eq:14} 
\frac{p}{2} = \frac{\lambda_+ + \lambda_-}{2} + \mu_+ d^2, \quad
\frac{1-p}{2} = \mu_- d^2, \\
\langle \psi_E, \sigma \psi_E \rangle = 
\lambda_+ +  \mu_+, \quad \langle \psi_{\tilde{E}}, \sigma \psi_{\tilde{E}} \rangle = 
\lambda_- +  \mu_+.  \label{eq:15}
\end{gather}
where $p$ is the probability from (\ref{eq:53}) and $d$ is given as above by:
$2^{m-1}$. 

%%%%%%%%%%%%%%%%%%%%%%%%%%%%%%%%%%%%%%%%%%%%%%%%%%%%%%%%%%%%%%%%%%%%%%%%%%%%%%%%
\section{A family of distillation protocols}
\label{sec:family-dist-prot}
%%%%%%%%%%%%%%%%%%%%%%%%%%%%%%%%%%%%%%%%%%%%%%%%%%%%%%%%%%%%%%%%%%%%%%%%%%%%%%%%

Let us consider now a large number of bipartite Fermionic
systems, each consisting of $2L$ modes\footnote{Basically, we could also use
an infinite number of modes at this point, but some of the statements made
in the last two sections are not valid in this case.}, which are
\emph{macroscopically distinguishable} (e.g. metallic wires, each  containing
an electron gas), and each prepared in the quasifree\footnote{Basically the
protocols we are going to describe work for non-quasifree states too, but
they are probably not very good. Furthermore, all non-trivial statements we will make
in this paper are related to quasifree states.} state $\rho_S$ (e.g. a
thermal equilibrium state or a ground state of the electron gas). To distill
from these systems maximally entangled pairs of $d$-level systems for some $d
\in \Bbb{N}$ we can proceed as follows (cf. \cite{K09}):

\begin{itemize}
\item \textbf{Step 1.} Choose an \emph{integer} $1 < m \leq L$ and a
\emph{real projection} $D \in \mathcal{B}(\mathcal{K}^{(L)})$ which commutes with
the projection $I_A$ onto Alice's modes, and satisfies $\tr(D_A) = \tr(D_B) =
m$. It defines a bipartite subsystem of our original system and we can drop
the remaining modes by applying the channel $\Delta_D$ defined in Equation 
(\ref{eq:51}). We get a Fermionic system consisting of $2m$ modes in the
quasifree state $\rho_{DSD}$ with covariance matrix $DSD \in
\mathcal{B}(\mathcal{K}^{(m)}) \cong \mathcal{B}(D\mathcal{K}^{(L)})$. The
purpose of this step is to discard low or unentangled parts of our original
system. The correct choice of $D$ requires knowledge about the state $\rho_S$.  
\item \textbf{Step 2.} Now choose a \emph{partial isometry} $V \in
  \mathcal{B}(\mathcal{K}^{(2m)})$, satisfying $VV^* = D_A$ and $V^*V =
  D_B$. It defines a maximally entangled, quasifree state $\psi_E \in
  \mathcal{H}^{(2m)}$ of the reduced system. It has to satisfy
  \begin{equation} \label{eq:3}
    \langle \psi_E, \rho_{DSD} \psi_E \rangle + \langle \psi_{\tilde{E}}, \rho
    \psi_{\tilde{E}}\rangle > \frac{p}{2^{m-1}},
  \end{equation}
  otherwise the protocol would fail. Now apply the corresponding twirling
  operation (\ref{eq:14}) to get the state $T_E(\rho_{DSD})$ from Equation
  (\ref{eq:7}). 
\item \textbf{Step 3.}  Make a joint parity measurement. This leads to a pair of
\emph{distinguishable systems} described by the tensor product of two
$d$-dimensional Hilbert spaces.  

  If the outcome is $++$ or $--$ (which  happens with probability $p$ given in
  Equation (\ref{eq:53})) this leads to the isotropic state $\sigma^{++}$ or
  $\sigma^{--}$ given by 
  \begin{equation} \label{eq:56}
    \sigma^{\pm\pm} = \frac{\lambda_+ + \lambda_-}{2 p^{\pm\pm}}
    |\psi_{E,\pm}\rangle\langle \psi_{E,\pm}| + \frac{\mu_+}{p^{\pm\pm}}
    P^{\pm\pm},
  \end{equation}
  where the maximally entangled states $\psi_{E,k} \in \mathcal{H}^{(m)}_k \otimes
  \mathcal{H}^{(m)}_k$ are defined by the decomposition (cf. Equation
  (\ref{eq:54})) 
  \begin{equation}
    \psi_E = \frac{1}{\sqrt{2}} \left( \psi_{E,+} + \psi_{E,-}\right)
  \end{equation}
  The fidelity $f = \langle \psi_{E,\pm}, \sigma^{\pm\pm} \psi_{E,\pm}\rangle$
  of $\sigma^{\pm\pm}$ is given by 
  \begin{equation} 
    f = \frac{\langle \psi_E, \rho_S \psi_E
      \rangle + \langle \psi_{\tilde{E}}, \rho_S
      \psi_{\tilde{E}}\rangle}{p}, \label{fidd} 
  \end{equation}
  such that Equation (\ref{eq:3}) becomes $f > \frac{1}{d}$. Hence
  $\sigma^{\pm\pm}$ is distillable. 

  Technically, $\sigma^{++}$ and $\sigma^{--}$
  are defined on different Hilbert spaces. Both, however, differ only by the
  value of the local parity operators $\theta \otimes \Bbb{1}$ and $\Bbb{1}
  \otimes \theta$. If we ignore the value of the parities (and only look
  whether they coincide or not) we can identify $\mathcal{H}^{(m)}_+ \otimes
  \mathcal{H}^{(m)}_+$ and $\mathcal{H}^{(m)}_- \otimes \mathcal{H}^{(m)}_-$
  such that $\psi_{E,+}$ and $\psi_{E,-}$, and $\sigma^{++}$ and $\sigma^{--}$
  become identical. 

  If the outcome of the measurement is $+-$ or $-+$ we get a totally mixed
  state which has to be discarded.
\item \textbf{Step 4.} Take the next system and restart at Step 1 (using the
same $D$ and the same $V$). Repeat this procedure
until enough systems in the state $\sigma^{\pm\pm}$ are available for
standard distillation techniques like hashing \cite{VoWo03} and proceed
accordingly. 
\end{itemize}

Hence, for each integer $n$ describing the size of the input systems, we have
defined a family of distillation protocols which is parametrized by the integer
$m$, the real projection $D$ and the partial isometry $V$. A useful choice of
all three parameters requires knowledge about the covariance matrix $S$ of the
input state $\rho_S$. The quality of the outputs is then measured by the
probability $p(S,m,D)$ from Equation (\ref{eq:53}) and the fidelity $f(m,D,V,S)$
given in (\ref{fidd}). Both quantities can be calculated easily according to the
following two theorems.

\begin{thm} \label{thm:3}
For each quasifree state $\rho_S$ of $2m$ Fermionic modes with covariance
matrix $S$ the probability $p = \tr(\rho_S P_+)$ to get the result $+$ in a
parity measurement  (cf. Equation (\ref{eq:53})) is given by  
\begin{equation} \label{eq:4}
p = \frac{1 + (-4)^m \Pf\bigl(-i(S - \Bbb{1}/2\bigr))}{2},
\end{equation}
where $\Pf$ denotes the \emph{Pfaffian}\footnote{Note that the Pfaffian can be
defined in a basis free way for any antisymmetric operator on a real,
even-dimensional, oriented Hilbert space. The $-i$ factor under the Pfaffian
in Equation (\ref{eq:4}) is necessary to get a real (rather than purely
imaginary) operator.}.
\end{thm}

\begin{proof}  
  See Appendix \ref{sec:parity-operator}.
\end{proof}

\begin{thm} \label{thm:4}
The fidelity of a quasifree state $\rho_S$ with covariance matrix $S$ with
respect to a quasifree, pure state $\psi_E$ with covariance matrix $E$ is
given by 
\begin{equation} \label{eq:25}
\langle\psi_E, \rho_S \psi_E\rangle = \Pf\bigl(-i(\Bbb{1} - S - E)\bigr).
\end{equation}
\end{thm}

\begin{proof}
  See \cite{KS10}.
\end{proof}

Hence the two quantities we are interested in become
\begin{equation} \label{eq:60}
p(S,m,D) =  \frac{1 + (-4)^m \Pf(DSD - \Bbb{1}/2)}{2} 
\end{equation}
\begin{multline} \label{eq:61}
f(S,m,D,V) = \frac{1}{2^{2m}} \Pf \left(
\begin{array}{cc}
D_AX_SD_A & D_AY_SD_B + V \\ -D_BY_S^*D_A -V & D_BZ_SD_B
\end{array}
\right) + \\ \frac{1}{2^{2m}} \Pf \left(
\begin{array}{cc}
D_AX_SD_B & D_AYD_B - V \\ -D_BY_S^*D_A +V & D_BZ_SD_B
\end{array}
\right),
\end{multline}
where we have used the fact that $X_{DSD} = D_A X_s D_A$ and similar relations
for $Y_S$ and $Z_S$ hold. 

Before we consider the question how to choose the
projection $D$ and the partial isometry $V$, let us have a short look on a
variant of the protocol given above. We can skip the twirl in Step 2 and
perform the joint parity 
measurement on the system in the state $\rho_{DSD}$. If the outcome is $+-$ or
$-+$ we drop the systems as before, and in the other case we get with
probability $p$ systems in the state
\begin{equation} \label{eq:75}
  \tilde{\rho} = \frac{p^{++}}{p} \rho_{DSD}^{++} + \frac{p^{--}}{p} \rho_{DSD}^{--},
\end{equation}
where we have identified (as above) the Hilbert spaces $\mathcal{H}^{(m)}_+
\otimes \mathcal{H}^{(m)}_+$ and $\mathcal{H}^{(m)}_- \otimes
\mathcal{H}^{(m)}_-$ such that $\psi_{E,+}$ and $\psi_{E,-}$ coincide. If we
twirl now with the symmetry group of $\psi_{E,\pm}$ we get the same output
states $\sigma^{\pm\pm}$ as in Step 3 above. In other words, nothing has
changed, apart from a reordering of steps. We will come back to this point,
however, in Section \ref{sec:two-four-modes} where we will compare our
protocol to a slightly different one.

%%%%%%%%%%%%%%%%%%%%%%%%%%%%%%%%%%%%%%%%%%%%%%%%%%%%%%%%%%%%%%%%%%%%%%%%%%%%%%%%
\section{The optimal protocol}
\label{sec:optimal-protocol}
%%%%%%%%%%%%%%%%%%%%%%%%%%%%%%%%%%%%%%%%%%%%%%%%%%%%%%%%%%%%%%%%%%%%%%%%%%%%%%%%

To distill entanglement from a given quasifree state $\rho_S$ with a protocol
from the family just introduced, the projection $D$ and the partial isometry
$V$ have to be chosen carefully, otherwise we do not get anything. The purpose
of this section is to show that there is a canonical choice, which provides
good or even optimal results. 

The crucial point is that any $S$ provides, by the \emph{polar decomposition}
\begin{equation}
Y_S = V_S |Y_S|
\end{equation}
of the off-diagonal block $Y = -2iS_{AB}$ (cf. Equation (\ref{eq:55})), a
partial isometry which is uniquely determined by the condition 
\begin{equation}
\Ran V_S = \Ran |Y_S| \subset \Ran I_B\quad  \text{and}\quad \ker V_S = \ker |Y_S|
\subset \Ran I_A.  
\end{equation}
Hence, any covariance matrix $S$ satisfying $\Ran |Y_S| = I_B$ and $\ker |Y_S| =
I_A$ (which holds iff $|Y_S|$ has maximal rank) defines a unique
quasifree, maximally entangled state via the partial isometry $V_S$. This
observation immediately suggests the following procedure to select $D$
and $V$:

\begin{itemize}
\item \textbf{Subsystem projection D.}
Consider an eigenbasis $e_k$, $k=1,\dots,4L$ of $|Y_S|$ such that the
corresponding eigenvalues are arranged in decreasing order, i.e.
\begin{equation} \label{eq:58}
|Y_S| e_k = \lambda_k e_k,\quad \lambda_1 \geq \lambda_2 \geq \dots \geq
\lambda_{4L} = 0. 
\end{equation}
Note that the kernel of $Y_S$ is by construction at least $2L$
dimensional. Hence $\lambda_{2L+1} = \cdots = \lambda_{4L} = 0$. Now choose
$2 \leq m \leq n$ small enough, but at least such that $\lambda_k > 0$
holds for all $k \leq 2m$ (we rediscuss strategies for the choice of $m$
below). With this preparation we can define $\hat{D} = \hat{D}_A \oplus
\hat{D}_B$ by
\begin{equation}
\hat{D}_B = \sum_{k=1}^{2m} \kb{e_k},\quad \hat{D}_A = V_S \hat{D}_B V_S^*. 
\end{equation}
This construction is in most cases basis-independent ($\hat{D}_B$ is just a
spectral projection of $|Y_S|$). The only exception arises if the eigenvalue
$\lambda_m$ is degenerate. In this case $\hat{D}_B$ depends on the choice of
a basis in the corresponding eingenspace.
\item \textbf{Partial isometry V.} Now consider the restricted state defined
by the covariance matrix $\hat{D}S\hat{D}$, and choose $\hat{V}$ such that 
\begin{equation} \label{eq:68}
Y_{\hat{D}S\hat{D}} = 
\hat{V} |Y_{\hat{D}S\hat{D}}|\quad \text{with}\quad
Y_{\hat{D}S\hat{D}} = -2i(\hat{D}S\hat{D})_{AB} = \hat{D}_A Y_S \hat{D}_B, 
\end{equation}
i.e. the maximally entangled state defined by $\hat{V}$ is the ``natural'' one
considered above.
\end{itemize}

We do not know yet, whether this choice really leads to a good distillation
rate, or whether it is even optimal. If we choose $m=2$ in Step 1 (i.e. the
output systems are qubit pairs) and hashing as the standard distillation
technique needed in Step 4 we get (where $S(\sigma^{\pm\pm})$ denotes the
entropy of the isotropic state $\sigma{\pm\pm}$ from Equation (\ref{eq:56})) 
\begin{align}
R &= p\max\bigl\{0,\bigl(1 - S(\sigma^{\pm\pm})\bigr)\bigr\}\\
&= p\max\bigl\{0, \bigl(1 + f\log_2(f) + (1-f)\log_2(1-f) -
(1-f)\log_2(3)\bigr) \bigr\}. \label{drate}
\end{align}
for the \emph{overall} distillation rate (cf \cite{VoWo03} for the rate of
hashing in any prime dimension). There are two problems with this
quantity. First of all if $S(\sigma^{\pm\pm}) > 1$, hashing does not lead to a
nonvanishing rate, although other strategies (like applying the recurrence
method first) might be more successful. Optimizing $R$ is pointless in these
cases. Second, even if $R > 0$ holds it is a very difficult function of $D$
and $V$ (cf. the expressions for $f$ and apply Theorems \ref{thm:3} and
\ref{thm:4}) and optimization is therefore a hard task. 

From a more general point of view, however, $R$ is just a product of $p$ and a
quantity which is monotone in $f$. This property it shares with 
\begin{equation}\label{eq:57}
pf = \langle \psi_E,\rho_S \psi_E \rangle + \langle \psi_{\tilde{E}},
\rho_S \psi_{\tilde{E}}\rangle, 
\end{equation}
which is independent of a special distillation strategy in Step 4 and easier
to optimize. As the rate $R$ it avoids strategies where the fidelity $f$ is
increased at the cost of the probability $p$ (which is possible, according to
Equation (\ref{fidd})). It shares with other fidelity quantities the drawback
that it can captures effects which are not related to entanglement at all
(e.g. local manipulation). In our case, however, and in particular for large
$f$, it provides a reasonably good figure of merit, which allows -- at least
for a large subclass of quasifree states $\rho_S$ -- the quantities $\hat{D}$
and $\hat{V}$ as optimal solutions. More precisely we have the following
theorem (the proof is given in the appendix).

\begin{thm} \label{thm:5}
Consider $2L$ Fermionic modes in a quasifree state $\rho_S$, an integer
$2 \leq m \leq L$, operators $D,V$, and the quantities $p(m,D,S)$ and
$f(m,D,V,S)$ from Equations (\ref{eq:60}) and (\ref{eq:61}). If $X_S = 0$
or $Z_S=0$ (cf. the notation from Equation (\ref{eq:55})), the product $pf$
satisfies
\begin{equation} \label{eq:59}
p(m,D,S) f(m,D,V,S) \leq p(m, \hat{D}, S) f(m, \hat{D}, \hat{V}, S) =
\prod_{k=1}^{2m} \frac{1+\lambda_k}{2} + \prod_{k=1}^{2m}\frac{1-\lambda_k}{2},
\end{equation}
where $\hat{D}$, $\hat{V}$ are defined above, and the $\lambda_k$ are the
$2m$ highest eigenvalues of $|Y_S|$ (cf. Equation (\ref{eq:58})).
\end{thm}

\begin{proof}
  See Appendix \ref{sec:optimality-proof}.
\end{proof}

Up to now we have kept the number  $m$ (i.e. the modes left after Step 1)
fixed. The reason is that a good choice for $m$ depends much more than the
choice of $D$ and $V$ on the initial state $\rho_S$ and the distillation
strategy used in Step 4. The crucial question is up to which degree it can be
more beneficial to increase the dimension $d$ of the output system at the
expense of some fidelity (cf. the corresponding discussion in
\cite{VoWo03}). If, however, we are only going for fidelity $m=2$ is the best
choice.

\begin{kor} \label{kor:1}
Consider the same assumptions as in Theorem \ref{thm:5}. Then we have
\begin{equation}
p(2, \hat{D}, S) f(2, \hat{D}, \hat{V}, S) \geq p(m, \hat{D}, S) f(m,
\hat{D}, \hat{V}, S). 
\end{equation}
\end{kor}

\begin{proof}
  See Appendix \ref{sec:optimality-proof}.
\end{proof}

The condition $X_S=0$ or $Z_S=0$ used so far is restrictive, but the corresponding 
class of states is still fairly big and grows quadratically in
$L$. Furthermore it seems to affect only local correlations (i.e. only between
Alice's modes or only between Bob's modes). Therefore
we conjecture that the protocol given by the operators $\hat{D}$, $\hat{V}$
presented above is generically a good if not the best choice within the family
introduced in Section \ref{sec:family-dist-prot}. We will refer to it in the
following as \emph{the optimal protocol}, even if the state to distill does
not satisfy the given condition.

%%%%%%%%%%%%%%%%%%%%%%%%%%%%%%%%%%%%%%%%%%%%%%%%%%%%%%%%%%%%%%%%%%%%%%%%%%%%%%%%%%%%%%%%%%%%%%%%%%%%
\section{Two and four modes}
\label{sec:two-four-modes}
%%%%%%%%%%%%%%%%%%%%%%%%%%%%%%%%%%%%%%%%%%%%%%%%%%%%%%%%%%%%%%%%%%%%%%%%%%%%%%%%%%%%%%%%%%%%%%%%%%%%

Now we will consider two simple systems, which can be treated explicitly in
the general case (without using the conditions of theorem
\ref{thm:5}). The first is one Fermionic mode for each party, which is only meant
to illustrate that the fidelity optimized over quasifree maximally entangled
states coincides (at least in this simple example) with that over any
maximally entangled state (i.e., the maximally entangled fraction). The second
is the first non-trivial case, two Fermionic modes for each party. Here we are
not looking at the optimization of the fidelity along the lines of Theorem
\ref{thm:5}, but we postpone, as already proposed at the end of Section
\ref{sec:family-dist-prot}, the twirling in Step 2 and do it only at the end
after the parity measurement. This leads to a two qubit system, whose maximally
entangled fraction can be explicitly calculated and therefore compared to the
fidelity our protocol provides.

In an appropriately chosen real basis $e_A^{1,2}, e_B^{1,2}$
(c.f. equation (\ref{rb})), the off-diagonal blocks of a general covariance
matrix $S$ are diagonal (cf. the corresponding discussion in Appendix
\ref{sec:optimality-proof}) and $S$ therefore becomes 
\begin{equation} 
S=\frac{1}{2}
\left(\begin{array}{rrrr}1&ia&ic&0\\-ia&1&0&id\\-ic&0&1&ib\\0&-id&-ib&1
\end{array}\right)\label{2qb}
\end{equation}
with $a,b,c,d\in{\mathbb R}$, and the constraint $0\leq S\leq  \Bbb{1}$ now
reads 
\begin{equation} 
1+(ab-cd)^2\geq a^2+b^2+c^2+d^2\leq 2\ .
\end{equation}
To get the corresponding density matrix it is convenient to use the
Jordan-Wigner transformation 
\begin{equation} 
B(e^i_A) =\frac{1}{\sqrt{2}}\,\sigma^i\otimes \Bbb{1},\quad
B(e^i_B) =\frac{1}{\sqrt{2}}\,\sigma^3\otimes  \sigma^i\quad i=1,2 .
\end{equation} 
By means of the Wick theorem we can then determine the correlation matrix
$r_{ij}=\rho_S(\sigma^i\otimes\sigma^j)$ where 
$\sigma^0=\Bbb{1}$. The matrix turns out to have the form
\begin{equation} 
r=\left(\begin{array}{rrrc}1&0&0&b\\0&0&d&0\\0&-c&0&0\\a&0&0&ab-cd
\end{array}\right)\ .
\end{equation}

Now consider a quasifree maximally entangled state $\psi_E$.
The corresponding basis projection is 
parametrized by the real orthogonal matrix given by equation (\ref{bpp}), 
which is now $2\times 2$. The fidelity $\langle \psi_E, \rho_S \psi_E\rangle$ maximized over 
all such $E$ then becomes
\begin{equation} \label{eq:43}
\sup_Y \langle \psi_{E},\rho_S\psi_{E}\rangle = \frac{1}{4}\max\{1-ab+cd+|c+d|,
1+ab-cd+|c-d|\}
\end{equation} 
where the supremum is attained by $Y = \sign (c+d)\Bbb{1}$ if the first and by $Y = \sign (c-d)
\diag(-1,1)$ if the second expression is larger. Note that the right hand side of (\ref{eq:43})
coincide with the supremum of $\langle\psi,\rho_S \psi\rangle$ over \emph{all} (i.e. not
necessarily quasifree) maximally entangled states \cite{BDiVSW}. This
indicates that we do not give away  distillation  quality by restricting in
Step 2 to maximally entangled \emph{quasifree} states.  

If we look at it as a Fermionic bipartite system, the two mode case is too simple, because the local
algebras $\mathcal{A}, \mathcal{B}$ (cf. Equations (\ref{eq:72}) and
(\ref{eq:74})) are Abelian -- hence no quantum degrees of freedom remain
\emph{locally}.  The easiest non-trivial case requires two mode on each side, hence four modes in
total. Unfortunately the corresponding covariance matrix $S$ has 16 free parameter (even after
diagonalizing the off-diagonal blocks) and it seems to be reasonable to make further simplifying
assumption. One possible choice is to assume (considering again an
appropriately chosen real basis) that the submatrix $Y_S$ is not just diagonal, but a
multiple of the identity, i.e. $Y_S = \sigma \Bbb{1}$ and that the commutator $[X_S,Z_S]$ vanishes. We can then
transform the submatrices $X_S$ and $Z_S$ to a normal form\footnote{An antisymmetric matrix can be always 
be transformed to a form of having only two-dimensional antisymmetric blocks in the diagonal by a 
unitary. When $X_S$ and $Z_S$ commute this can be done simultaneously by a Bogolubov unitary, 
obviously leaving  $Y_S = \sigma \Bbb{1}$ invariant. In this case, the state is a product of 2-qubit states.}  
%(cf. Proposition \ref{prop:2}) 
and $S$ becomes 
\begin{equation}
S=\frac{1}{2}\left(
\begin{array}{cc|cc|cc|cc}
1&i\nu_1&&&i\sigma&&&\\
-i\nu_1&1&&&&i\sigma&&\\\hline
&&1&i\nu_2&&&i\sigma&\\
&&-i\nu_2&1&&&&i\sigma\\\hline
-i\sigma&&&&1&i\nu_3&&\\
&-i\sigma&&&-i\nu_3&1&&\\\hline
&&-i\sigma&&&&1&i\nu_4\\
&&&-i\sigma&&&-i\nu_4&1
\end{array}\right) \ .\label{4qbs}
\end{equation}
In this case basically everything can be calculated explicitly, however, many of the results are
too complicated and not very useful. Therefore we will only give a summary of the most important
points. 

\begin{itemize}
\item 
  The probability to get the outcome $++$ or $--$ during a joint parity measurement is
  \begin{equation}
    p =  \frac{1 + (-4)^d \Pf(S - \Bbb{1}/2)}{2} = \frac{1+(\sigma^2-\nu_1\nu_3)
      (\sigma^2-\nu_2\nu_4)}{2}
  \end{equation}
\item 
  After such a measurement the (untwirled) posteriori states $\rho_S^{+-}$ and
  $\rho_S^{-+}$ are always separable, while $\rho_S^{++}$ and $\rho_S^{--}$ can be
  entangled. Their concurrence,  when non-vanishing, has the form
  $\sim\left(4\sigma^2-|h(\nu_1,\nu_2,\nu_3,\nu_4)|\right)$, which shows the
  expected behaviour:  $\sigma\neq 0$ is necessary for entanglement.
\item 
  The entanglement fidelity $f=f(S,2,\Bbb{1},V_S)$ from Equation  (\ref{eq:61})
  where $V_S$ is the partial isometry given by the polar decomposition $Y_S =
  V_S |Y_S|$ has the form (i.e. this is the output fidelity of the optimal
  protocol discussed in Section \ref{sec:optimal-protocol}) 
  \begin{equation}  
    f=\frac{(\nu_1\nu_3-\sigma^2-1)(\nu_2\nu_4-\sigma^2-1)+4\sigma^2}{8p}\ \label{fmodprot}.
  \end{equation}
\end{itemize}

The interesting question is now, whether the ``canonical'' choice $\hat{V} =
V_S$ for the partial isometry $V$ (and therefore the corresponding quasifree,
maximally entangled state $\psi_{\hat{E}}$) maximizes the quantity $pf$ or
what is equivalent in this case\footnote{We do not consider a subsystem,
  i.e. $m=2$ and the unit operator is the only possible choice for $D$. The
  probability $p$, however, does not depend on $V$, but only on $D$.} the
fidelity $f$. Direct optimization of the latter (e.g. if the isoclinic
decomposition of $\mathrm{SO}(4)$ is used to parameterize the set of
quasifree, maximally entangled states) can be done explicitly, but is quite
difficult (and therefore omitted here).

Fortunately, we can easily trace the problem back to a known result. To this
end recall the remark from the end of Section \ref{sec:family-dist-prot} that
we can basically swap the steps 2 and 3 of our protocol and do the joint
parity measurement directly on the state $\rho_S$. If the measured parities
coincides, and if we identify the Hilbert spaces $\mathcal{H}^{(2)}_+ \otimes
\mathcal{H}^{(2)}_+$ and $\mathcal{H}^{(2)}_- \otimes \mathcal{H}^{(2)}_-$ as
described in Section \ref{sec:family-dist-prot} we get the mixture
\begin{equation}
  \tilde{\rho} = \frac{p^{++}}{p} \rho_{S}^{++} + \frac{p^{--}}{p} \rho_{S}^{--},
\end{equation}
as the output state (cf. Equation (\ref{eq:75})). Twirling with the symmetry
group of $\psi_{E,+} = \psi_{E,-}$ leads to the same isotropic state we would
get if we run the protocol in the usual order (cf. again Section
\ref{sec:family-dist-prot}). This implies in particular that we can calculate
the fidelity $\langle \psi_{E,\pm}, \tilde{\rho} \psi_{E,\pm}\rangle$ in terms
of Equation (\ref{fidd}), i.e.
\begin{equation}
  \langle \psi_{E,\pm}, \tilde{\rho} \psi_{E,\pm}\rangle = \frac{\langle
    \psi_E, \rho_S \psi_E \rangle + \langle \psi_{\tilde{E}}, \rho_S
      \psi_{\tilde{E}}\rangle}{p}.
\end{equation}
The left hand side of this equation if bounded by the maximal singlet
fraction, i.e. the supremum of $\langle \Psi, \tilde{\rho}, \Psi\rangle$ over
\emph{all} maximally entangled two qubit states. The latter can be calculated
according to \cite{BDiVSW} which leads to
\begin{equation}  
\supp_{\Psi} \langle \Psi, \tilde{\rho} \Psi\rangle = \max\{f,g\}\ ,\label{mef}
\end{equation}
where $f$ is given in Equation (\ref{fmodprot}) and $g$ is
\begin{equation}  
g=\frac{(\nu_1\nu_3-\sigma^2+1)(\nu_2\nu_4-\sigma^2+1)}{8p}\ .
\end{equation}
Hence we have to determine the parameters for which $f$ is bigger than $g$,
because in this case the entanglement of the isotropic output states
$\sigma^{\pm\pm}$ is as big as it could be, and the optimal protocol proposed
in Section \ref{sec:optimal-protocol} is really optimal (within the given
class) although the assumption of Theorem \ref{thm:5}  ($X_S=0$ or $Z_S=0$) is
not met. 

\begin{figure}[h]\begin{center}\includegraphics[width=7cm]{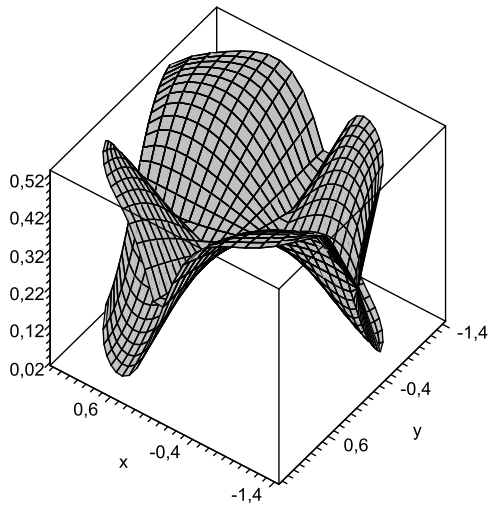}\end{center}{\caption
    {The figure shows two surfaces given by the fidelity $f$ of
      $\rho^4(x,y,\sigma=0.2)$ and its maximal singlet fraction
      $\max\{f,g\}$. Since $f>g$ for most of the values of the parameters, it is only at the left and right side, where we see that the two surfaces differ. This region shrinks if we increase $\sigma$ and disappears.}} 
\end{figure} 

We have done the last step numerically and chosen $x\equiv\nu_2\equiv\nu_1$
and  $y\equiv\nu_4\equiv\nu_3$. In this way we have a product of two identical
two-mode state $\rho^{(2)}(x, y,\sigma,\sigma)$ at hand as seen from
(\ref{4qbs}).  Using the notation $\rho^4(x,y,\sigma)$ for the state, we
compare the fidelity $f$ with the maximal entangled fraction (\ref{mef}). The
result is that $f \geq g$ for most of the parameter values, and in particular
for higher fidelities. 

%%%%%%%%%%%%%%%%%%%%%%%%%%%%%%%%%%%%%%%%%%%%%%%%%%%%%%%%%%%%%%%%%%%%%%%%%%%%%%%%
\section{Distillation from lattice Fermions}
\label{sec:dist-from-latt-1}
%%%%%%%%%%%%%%%%%%%%%%%%%%%%%%%%%%%%%%%%%%%%%%%%%%%%%%%%%%%%%%%%%%%%%%%%%%%%%%%%

To get another example for our scheme let us have a look at Fermions
propagating on an infinite, regular lattice. For simplicity, we will restrict
our discussion to one dimension, but everything we will say in this section,
can be easily generalized to higher dimensions.

The system we are looking at has infinite degrees of freedom. Therefore the
description used up to now has to be modified a  bit (cf
\cite{Araki87,MR0295702} for a detailed exposition). As a starting point
consider the Hilbert space $\mathcal{K}^{(\infty)}$ and the complex conjugation $J$ given
by 
\begin{equation} \label{eq:1}
  \mathcal{K}^{(\infty)} = \mathrm{l}^2(\Bbb{Z}) \oplus \mathrm{l}^2(\Bbb{Z}) \
  \text{and}\ \J(x,y) = (\overline{y}, \overline{x}).
\end{equation}
Now we can consider as before Majorana operators $B(x)$ smeared by elements 
$x \in \mathcal{K}^{(\infty)}$, operating irreducibly on a (now infinite
dimensional, but separable) Hilbert space $\mathcal{H}^{(\infty)}$ and
satisfying the CAR. As in finite dimensions, a normalized element $\Omega_S \in
\mathcal{H}^{(\infty)}$ describes a pure, quasifree state if its correlation
functions are given in terms of Wick's Theorem (cf. Equation (\ref{eq:38}))
and a basis projection $S \in \mathcal{B}(\mathcal{K}^{(\infty)})$, which is
defined as in the finite dimensional case. The Hilbert space
$\mathcal{H}^{(\infty)}$ and the operators $B(x)$ can be 
constructed as in Appendix \ref{sec:self-dual-formalism} in terms of Fock
spaces and creation/annihilation operators, but this is not important for us,
because the triple consisting of $\mathcal{K}^{(\infty)}$, $J$ and $S$ fixes
$\mathcal{H}^{(\infty)}$ and the $B(x)$ uniquely up to unitary equivalence.

Consider now a finite set $\Lambda$ of integers with cardinality $n$ and
define the projection $I_\Lambda \in \mathcal{B}(\mathcal{K}^{(\infty)})$
\begin{equation}
  (I_\Lambda x)_j = \chi_\Lambda(j) x_j,\quad j \in \Bbb{Z}
\end{equation}
where $\chi_\Lambda$ denotes the characteristic function of $\Lambda$. This
projection commutes with $J$ (i.e. it is a \emph{real} projection) and defines 
therefore a (finite dimensional) subsystem of the infinitely extended
system. It is completely described by the algebra
\begin{equation}
  \mathcal{A}_\Lambda = \SP \{ B(x_1) \cdots B(x_l) \, | \, l \in \Bbb{N},\
  I_\Lambda x_j = x_j \forall j=1,\dots,l \}
\end{equation}
which is isomorphic to the CAR algebra of $n$ modes. All operations,
measurements, etc., which only affect the modes located in $\Lambda$  (and leave
the rest of the system untouched) can be described in terms of this
algebra. It is acting reducibly on $\mathcal{H}^{(\infty)}$, but we can cure
this defect with a twisted tensor product. In analogy to (\ref{eq:48}) we can
introduce a unitary  
\begin{equation}
  U_\Lambda : \mathcal{H}^{(\infty)} \rightarrow \mathcal{H}^{(n)} \otimes
  \mathcal{H}^{(\infty)} 
\end{equation}
such that 
\begin{equation}
  U_\Lambda B(x) U_\Lambda^* = 
  \begin{cases}
    B(x) \otimes \Bbb{1} &  x \in I_\Lambda \mathcal{K}^{(\infty)} \cong
    \mathcal{K}^{(n)} \\ 
    \theta \otimes B(x) & x \in (\Bbb{1} - I_\Lambda) \mathcal{K}^{(\infty)}
    \cong \mathcal{K}^{(\infty)}, 
  \end{cases}
\end{equation}
where $\mathcal{H}^{(n)}$ denotes, as before, the Fermionic Fockspace of $n$
modes. Dropping the complementary subsystem (which is now infinite
dimensional) can be described again by the channel (cf. Equation (\ref{eq:51})
\begin{equation} \label{eq:63}
  \mathcal{B}(\mathcal{H}^{(n)}) \ni A \mapsto \Delta_\Lambda(A) = U_\Lambda^*
  A \otimes \Bbb{1} U_\Lambda \in \mathcal{B}(\mathcal{H}^{(\infty)}).
\end{equation}
After applying this operation, the remaining $n$ modes are in the quasifree
state $\rho_{S,\Lambda}$ with covariance matrix $S_\Lambda = I_\Lambda S
I_\Lambda$ and we end up with the structure studied in Section
\ref{sec:math-prel}. 

Now let us come back to distillation and assume that Alice and Bob share an
infinite 1-dimensional lattice containing a number (typically infinite) of
Fermions. They can control only two contiguous regions $\Lambda_A$
and $\Lambda_B$ of $L$ lattice sites each and located at a distance of $N$
sites.
Everything outside of $\Lambda = \Lambda_A \cup \Lambda_B$ can not be
manipulated and is therefore ignored\footnote{We ignore here the fact actions
  by Alice and Bob on their part of the system usually produces disturbances
  which are propagating along the chain. This approximation is justified if
  the time scale at which Alice and Bob operate is short compared to the
  propagation speed.}. In other words we apply the channel $\Delta_\Lambda$
defined in Equation (\ref{eq:63}) and get a bipartite Fermionic system
consisting of $2L$ modes where the projections $I_{\Lambda_A}$ and
$I_{\Lambda_B}$ play the role of $I_A$ and $I_B$ from Section
\ref{sec:bipart-syst-entangl}. Hence we can apply the optimal protocol studied
in Section \ref{sec:optimal-protocol} and calculate the success probability
$p(L,N)$ and the fidelity $f(L,N)$ of the output systems. Both quantities
depend on the geometry of the regions $\Lambda_A$ and $\Lambda_B$ (i.e. their
length $L$ and distance $N$) and to study this functional behavior gives a
deep insight how distillable entanglement is distributed along the chain.

Particularly interesting is the question, whether $p$ and $f$ can be made (as a
function of $L$) arbitrarily close to 1. This would imply that using a
sufficiently big chunk of the system is as good as using a sufficiently large
number of the same system in the same state (which is typically done in
ordinary distillation protocols such as hashing, and we have assumed this in
Section \ref{sec:family-dist-prot}, too). In this sense our result is closely
related to the question how much entanglement can be distilled with certainty
from \emph{one copy} of an infinitely extended system. This is called the
systems \emph{one copy entanglement} and studied in a number of papers
\cite{InfEnt,Eisert05a,KMSW06}.  

%%%%%%%%%%%%%%%%%%%%%%%%%%%%%%%%%%%%%%%%%%%%%%%%%%%%%%%%%%%%%%%%%%%%%%%%%%%%%%%%
\section{Free Fermions hopping on a 1-dimensional lattice}
\label{sec:free-ferm-hopp-1}
%%%%%%%%%%%%%%%%%%%%%%%%%%%%%%%%%%%%%%%%%%%%%%%%%%%%%%%%%%%%%%%%%%%%%%%%%%%%%%%%

As an explicit example for the discussion from the last section let us consider the unique ground
state of free Fermions hopping on a infinite, 1-dimensional, regular lattice\footnote{After a
  Jordan-Wigner transformation this becomes the XX-model. One copy entanglement of this model is
  studied in \cite{Eisert05a}; cf. also the references therein for other related results.}
(i.e. $\Bbb{Z}$). It is given by
\begin{equation}
  S= \scr{F}^{-1} \left( 
    \begin{array}{cc}
      E & 0 \\
      0 & 1-E
  \end{array} \right) \scr{F}
\end{equation}
where 
\begin{equation}
  \scr{K}^{(\infty)} \ni F \mapsto \scr{F}(F) \in \Lz(S^1) \otimes \Bbb{C}^2, \quad 
   \scr{F}(F)(x) = \frac{1}{\sqrt{2\pi}}\sum_{j=-\infty}^\infty e^{inx} F_n
\end{equation}
is the Fourier transform and $E \in \scr{B}(\Lz(S^1))$ the projection to the upper half-circle. Proofs
and further references for all of this can be found in \cite{MR810491}.

As in the last section we will restrict the state to the region $\Lambda$ with 
\begin{equation}
  \Lambda_A = [-L,0),\quad \Lambda_B = [N,N+L),\quad \Lambda = \Lambda_A
  \cup \Lambda_B.
\end{equation}
This results in the restricted density matrix $\rho_{I_\Lambda S I_\Lambda}$, with a covariance
matrix $I_\Lambda S I_\Lambda$ which becomes in an appropriately chosen real
basis 
\begin{equation}
  I_\Lambda S I_\Lambda = \left( 
    \begin{array}{cc}
      \Bbb{1}_A/2 +i X & i Y \\
      -i Y^T & \Bbb{1}_B/2 + i Z
    \end{array} \right)
\end{equation}
with
\begin{equation} \label{eq:12}
  X = Z =
  \left(\begin{array}{cc}
    0 & i F_0 \\
    -i F_0 & 0
  \end{array}\right),
  \quad Y = 
  \left(\begin{array}{cc}
    0 & i F_{N+L} \\
    -i F_{N+L} & 0
  \end{array}\right)
\end{equation}
and the $L \times L$ matrix $F_r$ ($r=0,N$) is given by
\begin{equation}
  (F_r)_{jk} =
  \begin{cases}
    0 & \text{if}\ j=k+r \\ 
    \frac{\sin\left(\frac{(j-k+r)\pi}{2}\right)}{(j-k+r) \pi} & \text{otherwise.}
  \end{cases}
\end{equation}

Our goal is now to consider the distillation of qubit pairs from bipartite systems in such a
state. The crucial step on the algorithmic side is to determine the subsystem
projection $\hat{D}$ and the partial isometry $\hat{V}$ from $I_\Lambda S I_\Lambda$. This is done
in terms of the singular value decomposition
\begin{equation}
  Y = \sum_{j=1}^{2L} \lambda_j \KB{\phi_j}{\Phi_j}, \quad \lambda_1 \geq
  \lambda_2 \geq \dots \lambda_{2L}
\end{equation}
of the matrix $Y$ from Equation (\ref{eq:12}). This gives rise to 
\begin{equation}
  \hat{D}_A = \sum_{j=1}^4 \kb{\phi_j},\quad \hat{D}_B = \sum_{j=1}^4
  \kb{\Phi_j}, \quad \hat{V} = \sum_{j=1}^4 \KB{\phi_j}{\Phi_j}.
\end{equation}
All other quantities (in particular the probability $p$ and the fidelity $f$)
can be derived from these three matrices (cf. Theorem \ref{thm:3} and
\ref{thm:4}). 

Unfortunately, even for the matrices $F_0,F_{N+L}$ with their very regular structure the singular value
problem is not explicitly solvable. They are, however, Toeplitz matrices and we can at least provide
an algorithm to derive $D_\Lambda S D_\Lambda$ which is quite efficient (in time and space) and
which allows to do numerical calculations (at least) in the range $L=10^4 - 10^6$. One possible
strategy (and the one we have used) is to use iterative procedures like the L\'anczos method, which
only needs matrix-vector multiplications. The latter can then be implemented by embedding the
Toeplitz matrix $T$ into a bigger circulant matrix $C$ and to calculate the product $Cx$ for a
column vector $x$ by a discrete Fourier transform. The product $Ty$ can then be calculated by
first padding $y$ with zeros to fit the size of $C$ (which results in a higher dimensional vector
$x$) and projecting out at the end the relevant dimensions from $Cx$. 

\begin{figure}[htbp]
  \centering
  \includegraphics[scale=.8]{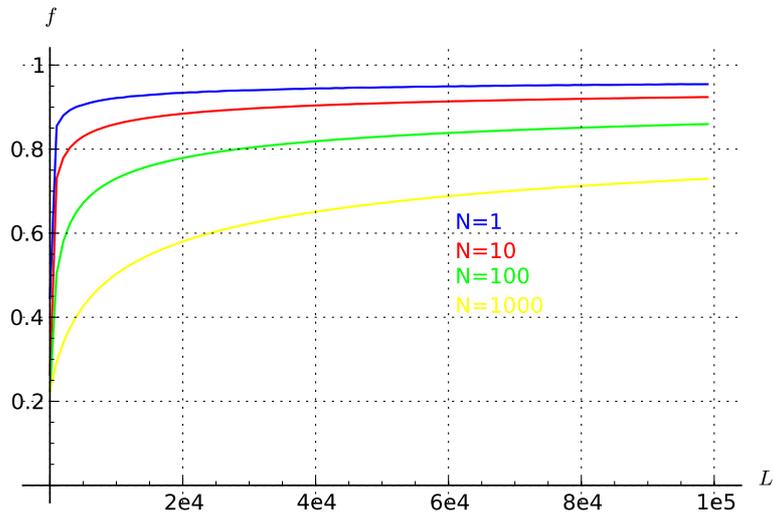}
  \caption{\label{fig:1} Plot of the entanglement fidelity $f$ as a function of the block
    length $L$ for fixed block distances $N=0,10,100,1000$.}
\end{figure}

\begin{figure}[htbp]
  \centering
  \includegraphics[scale=.8]{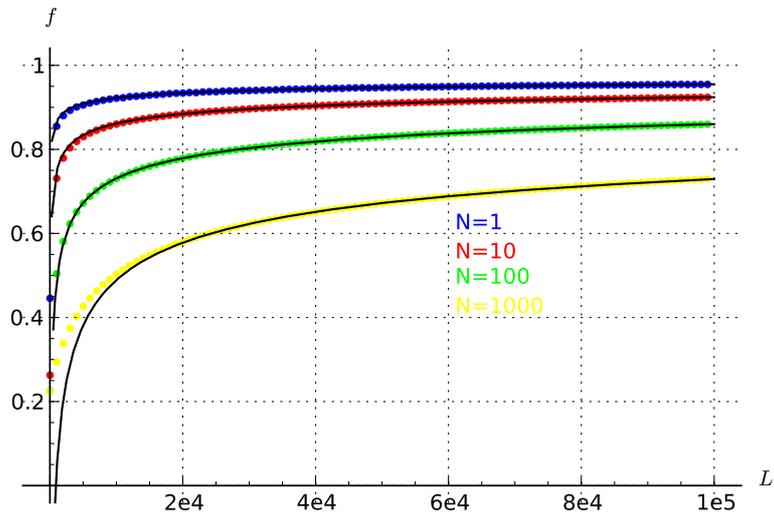}
  \caption{\label{fig:2} Entanglement fidelity $f(L)$ for $N=1,10,100,1000$ fitted with a
    model function (black curves) of the type $1-b/L^a$.}
\end{figure}

% We have applied this procedure to the distillation of entangled qubit pairs according to Protocol B
% or C. The entanglement fidelity $f$ of the isotropic states $\sigma^{\pm\pm}$ we get in Step 4B
% (which is identical with Step 4A) or of the state $\rho^{\pm\pm}_{D_\Lambda S D_\Lambda}$ from Step
% 3C is given by
% \begin{equation}
%   f = \frac{\langle\psi_P, \rho_{D_\Lambda S D_\Lambda} \psi_P\rangle + \langle\psi_Q,
%     \rho_{D_\Lambda S D_\Lambda} \psi_Q\rangle}{p} 
% \end{equation}
% with basis projections $P,Q$ from Equations (\ref{eq:26}) and (\ref{eq:10}) respectively and with
% the probability $p$ to get $++$ or $--$ during a joint parity measurement. All quantities can be
% easily calculated from the reduced covariance matrix $D_\Lambda S D_\Lambda$ which can be calculated
% along the lines described in the last paragraph. 

\begin{figure}[htbp]
  \centering
  \includegraphics[scale=.8]{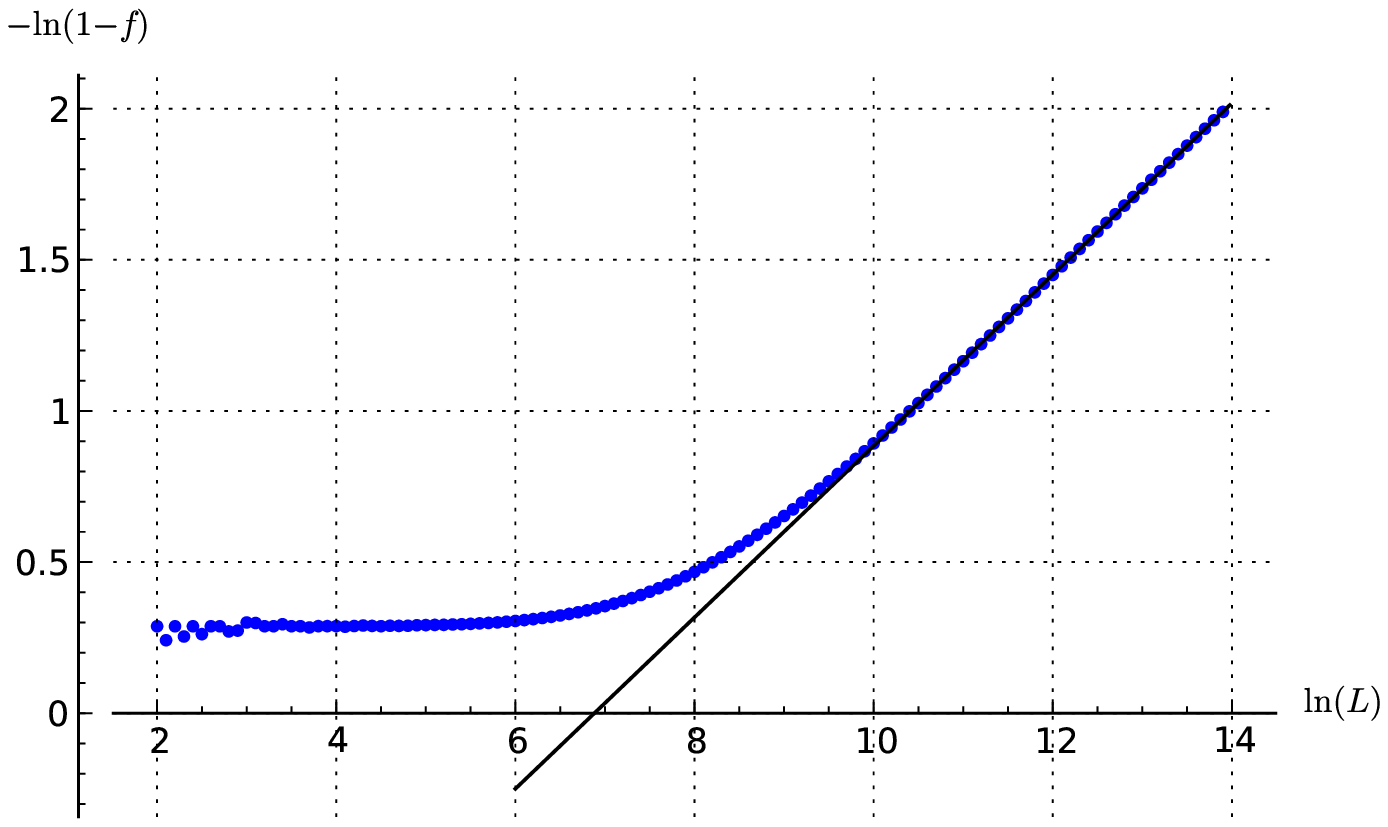}
  \caption{\label{fig:3} $-\log(1-f)$ as a function of $\log(L)$ for $N=1000$. For large $L$ the functional
    behavior seems to be linear (cf. the black curve).}
\end{figure}

Let us discuss now the result of our analysis. Figures \ref{fig:1} and
\ref{fig:2} show $f$ as a function of $L$ for fixed block distances
$N=1,10,100,1000$. In Figure \ref{fig:2} we have fitted the data with a model
function of the form 
\begin{equation}
  f(L,N) = 1 - \frac{b(N)}{L^{a(N)}}.
\end{equation}
It is easy to see that for small values of the block length $L$ this fit is not very good (in
particular for larger block distance $N$). Therefore we have used only data $(L,f(N,L))$ with $L\geq
20000$ for calculating the fit parameters $a$ and $b$ (using a least square fit). This difference
between the behavior for small and large $L$ becomes much clearer in Figure \ref{fig:3} where we
have plotted $-\log(1-f)$ against $\log(L)$. The linear part for bigger values of $\log(L)$
correspond to the $1-b/L^a$ behavior in Figure \ref{fig:2}. 
\begin{figure}[htbp]
  \centering
  \includegraphics[scale=.8]{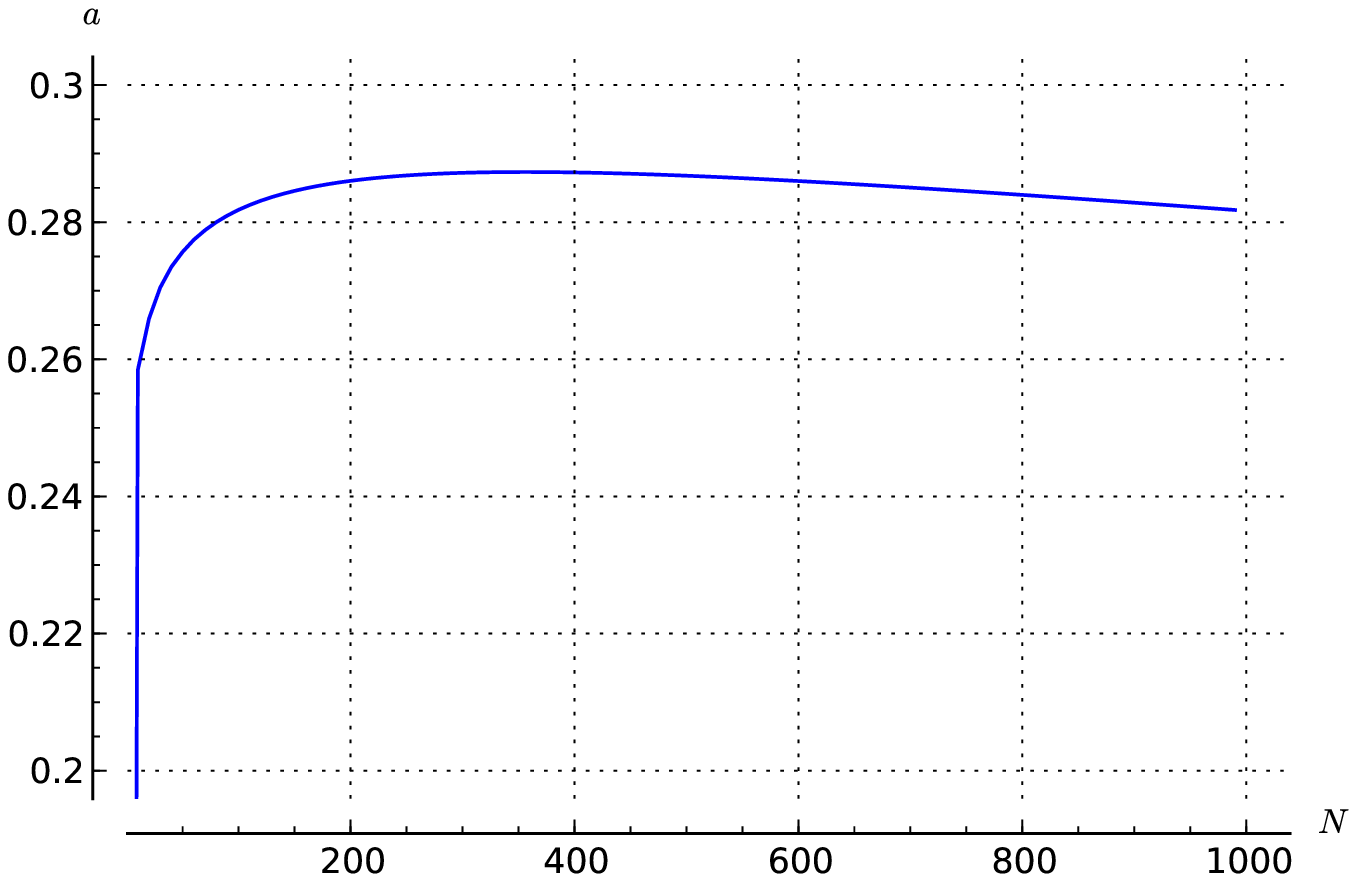}
  \caption{\label{fig:4} Fit parameter $a$ as a function of $N$. The black curve represents the
    average of the values for $N \geq 400$.}
\end{figure}

\begin{figure}[htbp]
  \centering
  \includegraphics[scale=.8]{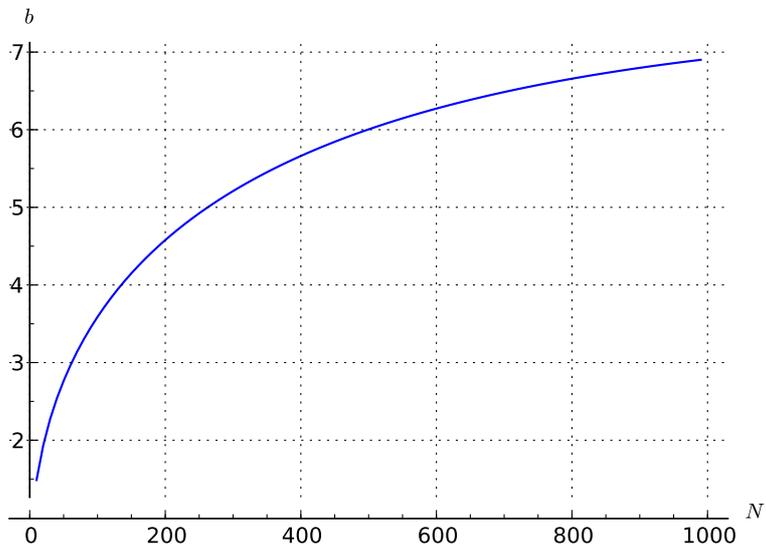}
  \caption{\label{fig:5} Fit parameter $b$ as a function $N$. For $N>500$ the data is fitted
    logarithmically.}
\end{figure}

To get an idea how the two fit parameter $a$ and $b$ depend on $N$ we have made a least square fit
for the data $(L,f(L,N))$ with $L = 10 + 1000k$, $k \in [20,100)$ and for all $N = 0 + 10n$ with
$n \in [0,100)$. The results are shown in Figures \ref{fig:4} and \ref{fig:5}. We can see that for
larger values of $N$ the parameter $a$ is almost constant, while $b$ seems to grow
sub-linearly. Note that it is dangerous to say more at this point, because small fluctuations can be
a consequence of the fitting algorithm rather than a real physical effect.

More insight we can get from a look on the function $L(N,x)$ which provides the minimal length
of the intervals $\Lambda_{A/B}$ we need to get\footnote{Note that this is the quantity proposed to
  be studied in \cite{KMSW06}.} $f(L(N,x),N)\geq x$. The result is shown in Figure \ref{fig:6}. We
can see that the behavior looks linear and close inspection of the numerical data indeed shows that
at least for large values of $N$ the behavior is as linear as it can be for a function which only
takes integer values.

\begin{figure}[htbp]
  \centering
  \includegraphics[scale=.8]{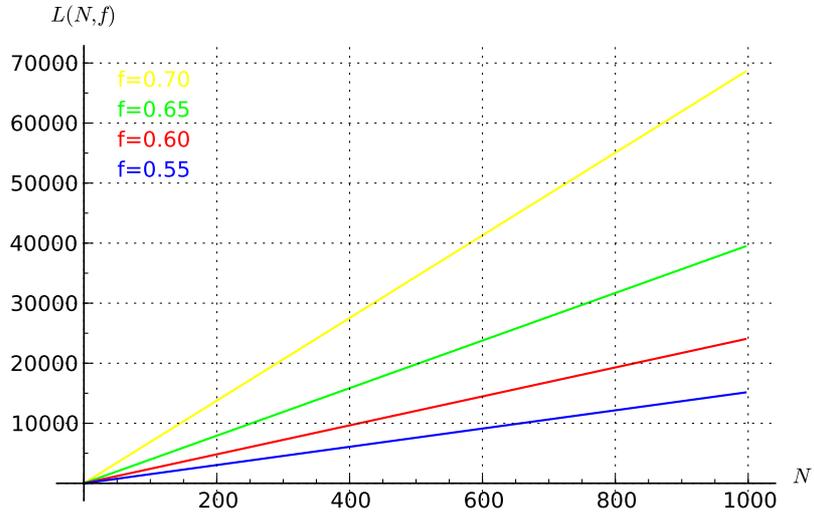}
  \caption{\label{fig:6} Minimal interval length $L(N,x)$ needed to get a fidelity $f\geq x$
  if the distance is $N$.} 
\end{figure}

Another important function we can investigate is the probability $p$ to get during a joint parity
measurement on the restricted system in the state $\rho_{D_\Lambda S D_\Lambda}$ the results $++$ or
$--$. 

\begin{figure}[htbp]
  \centering
  \includegraphics[scale=.8]{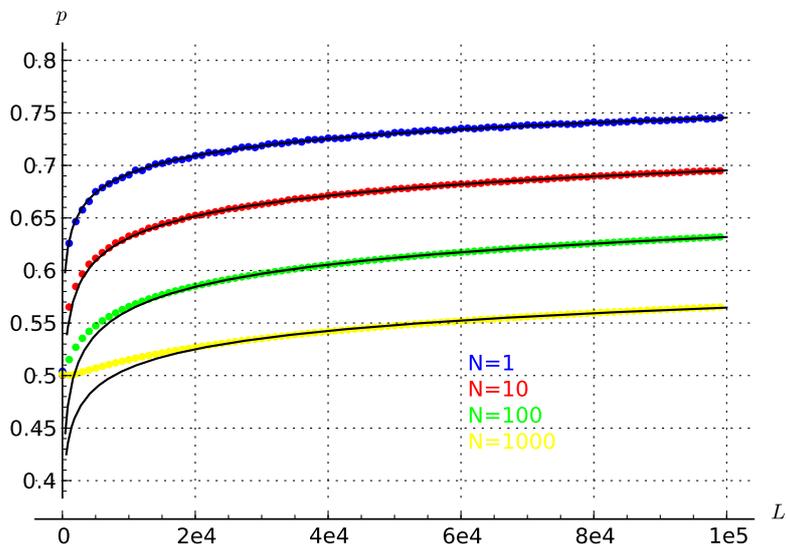}
  \caption{\label{fig:7} Probability $p(L,N)$ for $N=1,10,100,1000$ fitted with a
    model function (black curves) of the type $1-b/L^a$.}
\end{figure}

\begin{figure}[htbp]
  \centering
  \includegraphics[scale=.8]{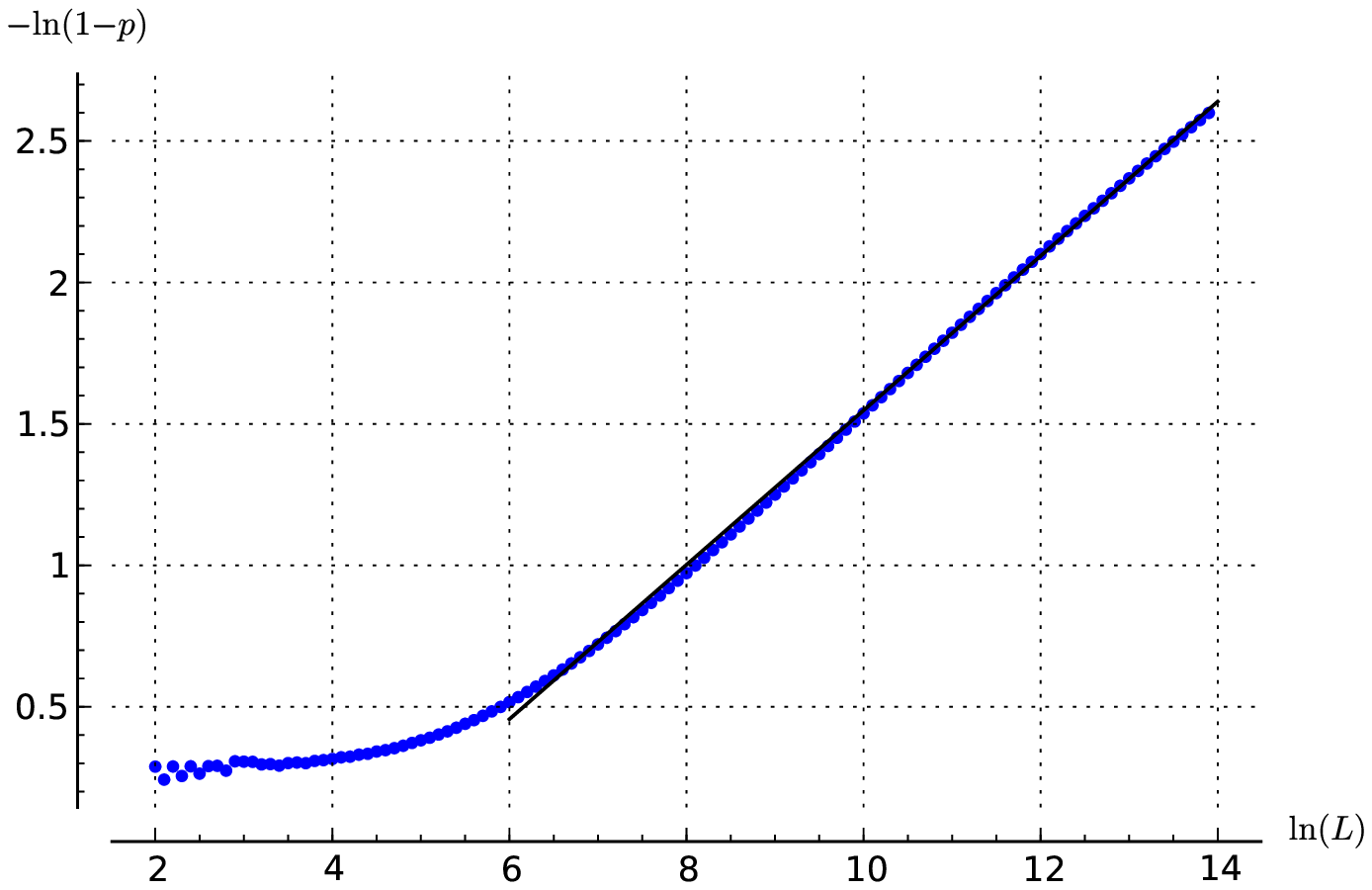}
  \caption{\label{fig:8} $-\log(1-p)$ as a function of $\log(L)$ for $N=1000$. For large $L$ the functional
    behavior is linearly approximated.}
\end{figure}

In Figure \ref{fig:7} we can see that the behavior of $p$ as a function of $L$ looks quite similar
as $f(L)$ from Figure \ref{fig:2}. Again we can find a good fit for larger $L$ with a function of
the type $1-a/L^b$. This is also confirmed by Figure \ref{fig:8} where we have again plotted
$-\log(1-p)$ as function of $\log(p)$. The behavior of $p$ for small $L$ and $N$ seems, however, to
be more irregular than the one for $f$ and if we look very closely at Figure \ref{fig:8} we can see
that the linear fit is not as good as in Figure \ref{fig:3} (the blue curve is bending a little bit
around the black line). 

%%%%%%%%%%%%%%%%%%%%%%%%%%%%%%%%%%%%%%%%%%%%%%%%%%%%%%%%%%%%%%%%%%%%%%%%%%%%%%%%%%%%%%%%%%%%%%%%%%%%
\section{Conclusions}
%%%%%%%%%%%%%%%%%%%%%%%%%%%%%%%%%%%%%%%%%%%%%%%%%%%%%%%%%%%%%%%%%%%%%%%%%%%%%%%%%%%%%%%%%%%%%%%%%%%%

We have presented a distillation scheme which can be applied to a bipartite Fermionic system in a
generic quasifree state, and while it is in general not optimal we have found substantial
evidence that we do not waste too much entanglement originally hidden in the given state. The
asymptotic rate of this protocol (or equivalently the entanglement fidelity of the isotropic states
we get at an intermediate step) can be calculated explicitly and can therefore be used as tool to
analyze entanglement of Fermionic systems.

The latter is particularly true for lattice Fermions. The analysis of the previous section showed
that a numerical study is possible for effectively very large systems up to $10^6$ lattice
sites. We have only studied a very special quasifree state here, but the algorithms used can be
easily generalized to other translation invariant quasifree states: Their covariance operator is
a ($2 \times 2$-matrix valued) multiplication operator and the matrix $I_\Lambda
S I_\Lambda$ of the restricted state (which can be easily calculated by
Fourier integrals) is again a (block-)Toeplitz operator. 

The results about free Fermions on a lattice, derived in the last section  are interesting in their
own right and indicate some structure in the fidelity $f(N,L)$ and the block length $L(N,f)$. Note
that the latter is closely related to a quantity proposed in \cite{KMSW06} to analyze states with
infinite one-copy entanglement. The numerical study presented so far does not tell us the whole
story of course, but it gives us some hints what we might expect from a more detailed, analytic
study. The latter is of course difficult but a deeper discussion of the asymptotic behavior of $f$
and $L$ (which is related to the asymptotic behavior of the singular value decomposition of the
off-diagonal blocks $Y$) seems to be a realistic goal. 

%%%%%%%%%%%%%%%%%%%%%%%%%%%%%%%%%%%%%%%%%%%%%%%%%%%%%%%%%%%%%%%%%%%%%%%%%%%%%%%%%%%%%%%%%%%%%%%%%%%%
\section*{Acknowledgement}
%%%%%%%%%%%%%%%%%%%%%%%%%%%%%%%%%%%%%%%%%%%%%%%%%%%%%%%%%%%%%%%%%%%%%%%%%%%%%%%%%%%%%%%%%%%%%%%%%%%%

This work was supported by the EU FP7 FET-Open project COQUIT (contract number 233747).

\begin{appendix}

%%%%%%%%%%%%%%%%%%%%%%%%%%%%%%%%%%%%%%%%%%%%%%%%%%%%%%%%%%%%%%%%%%%%%%%%%%%%%%%%
\section{The self-dual formalism}
\label{sec:self-dual-formalism}
%%%%%%%%%%%%%%%%%%%%%%%%%%%%%%%%%%%%%%%%%%%%%%%%%%%%%%%%%%%%%%%%%%%%%%%%%%%%%%%%

The purpose of this appendix to connect the self-dual formalism presented in
Section \ref{sec:math-prel} to the more familiar description in terms of
creation and annihilation operators $c_j,c_j^*$, $j=1,\dots,n$ operating on
Fermionic  Fockspace $\mathcal{H}^{(n)} = \mathcal{F}_-(\Bbb{C}^n)$. They
satisfy the canonical anticommutation relations (CAR) 
\begin{equation} \label{eq:19}
  \{c_j,c_k\} = \{c_j^*,c_k^*\} = 0,\quad \{c_j,c_k^*\} = \delta_{jk} \Bbb{1}
\end{equation}
and they are irreducible, i.e. for any operator $A$ on $\mathcal{H}^{(n)}$ we
have 
\begin{equation}
  [A,c_j] = [A,c_k^*] = 0 \ \forall j,k = 1,\dots,n \Rightarrow A = \lambda \Bbb{1}
\end{equation}
for some $\lambda \in \Bbb{C}$. The latter implies in particular that the
*-algebra generated by the $c_j, c_k^*$ coincides with
$\mathcal{B}(\mathcal{H}^{(n)})$. 

Alternatively we can introduce the Majorana operators
\begin{equation}
  B_a = 
  \begin{cases}
    \frac{1}{\sqrt{2}} (c_a + c_a^*) & \text{if $a = 1,\dots,n$}\\
      \frac{i}{\sqrt{2}} (c_{a-n} - c_{a-n}^*) & \text{if $a=n+1,\dots,2n$}.
  \end{cases}
\end{equation}
The $B_a$ are selfadjoint and from (\ref{eq:19}) we immediately get
\begin{equation} \label{eq:44}
  \{ B_a, B_b \} = \delta_{a,b} \Bbb{1}.
\end{equation}
Irreducibility of the $c_j,c_k^*$ implies irrducibility of the $B_a$ and vice
versa. 

Finally we can smear the $B_a$ with a ``test function'' $x \in \Bbb{C}^{2n}$
\begin{equation}
 B(x) = \sum_{a=1}^{2n} x_a B_a,\quad x = (x_a)_{a=1,\dots,2n} \in
 \mathcal{K}^{(n)} = \Bbb{C}^{2n}
\end{equation}
to get the structure we have started with in Section \ref{sec:math-prel}. Note
that in this case we have $\mathcal{K}^{(n)} = \Bbb{C}^{(2n)}$ and $J$ is
given by complex conjugation in the canonical basis ($(Jx)_a = \bar{x}_a$).
It is easy to see that the $B(x)$ are irreducible, linear in $x$ and satisfy
the CAR in the form (\ref{eq:76}).

If we start with the smeared operators we can derive the $B_a$ by $B_a =
B(e_a)$ if $e_a$ denotes the canonical basis in $\mathcal{K}^{(n)}$. Instead
of the $e_a$, however, we can use \emph{any} real basis $\tilde{e}_a, a=1,
\dots, 2n$ of $\mathcal{K}^{(n)}$ to get an irreducible  family of selfadjoint
operators $\tilde{B}_a = B(\tilde{e}_a)$ on $\mathcal{H}^{(n)}$, which satisfy
the CAR (\ref{eq:44}) and are therefore as good to describe the system as the
orginal $B_a$. In other words: The smeared operators $B(x)$ provide a
coordinate free description of the system, while the $B_a$ (and the
corresponding $c_j,c_k^*$) are coordinate dependent.  

This distinction is not purely academic, because the $c_j,c_k^*$ (and
therefore the $B_a$ from which the $c_j,c_k^*$ can be derived) are accompanied
by a distinguished state: The Fock vacuum which is given by the condition
$c_k^* \Omega = 0$ $\forall k=1,\dots,n$. The smeared operators $B(x)$,
however, does not prefer \emph{any} state -- we have to choose a real basis
first. 

To make the last statement more precise let us show how the definition of Fock
states in Section \ref{sec:math-prel} is related to the ordinary Fock
vacuum. To this end let us consider a basis projection $E \in
\mathcal{B}(\mathcal{K}^{(n)})$ and introduce orthonormal bases
$(f_j^{(1)})_{j=1,\dots,n}$ and $(f_k^{(2)})_{k=1,\dots,n}$  in the range
$\Ran E \subset \mathcal{K}^{(n)}$ and kernel $\ker E \subset
\mathcal{K}^{(n)}$ of $E$ respectively. Obviously  
\begin{equation}
  e_a = \begin{cases}
    \frac{1}{\sqrt{2}} (f_a^{(1)} + f_a^{(2)}) & \text{if $a = 1,\dots,n$}\\
      \frac{i}{\sqrt{2}} (f_{a-n}^{(1)} - f_{a-n}^{(2)}) & \text{if $a=n+1,\dots,2n$}.
  \end{cases}\label{rb}
\end{equation}
is a real basis and a short calculation shows that the vacuum state $\Omega
\in \mathcal{H}^{(N)}$ given by 
\begin{equation}
  B_a = B(e_a),\quad c_k = \frac{1}{\sqrt{2}} (B_k + i B_{k+n}),\quad c_k^*
  \Omega = 0\ \forall k=1,\dots,n
\end{equation}
satisfies
\begin{equation}
  \langle \Omega, B(x_1) B(x_2) \Omega \rangle = \langle J x_1, E x_2 \rangle.
\end{equation}
Validity of the relations in (\ref{eq:38}) follows from Wick's
theorem. Therefore $\Omega \in \mathcal{H}^{(n)}$ is the Fock state with
covariance operator $E$.

%%%%%%%%%%%%%%%%%%%%%%%%%%%%%%%%%%%%%%%%%%%%%%%%%%%%%%%%%%%%%%%%%%%%%%%%%%%%%%%%%%%%%%%%%%%%%%%%%%%%
  \section{The parity operator}
  \label{sec:parity-operator}
%%%%%%%%%%%%%%%%%%%%%%%%%%%%%%%%%%%%%%%%%%%%%%%%%%%%%%%%%%%%%%%%%%%%%%%%%%%%%%%%%%%%%%%%%%%%%%%%%%%%

  In this appendix we will present a proof of Theorem \ref{thm:3} and Equation
  (\ref{eq:50}). Unless something else is explicitly stated we will use the
  notations from Section \ref{sec:math-prel}. The core result is the following proposition.

  \begin{prop}
    Let $e_a \in \mathcal{K}^{(n)}$, $a=1,\dots,2n$ be a real basis of the
    reference space $\mathcal{K}^{(n)}$ and define
    \begin{equation}
      \theta = 2^n i^n B(e_1) \dots B(e_{2n});
    \end{equation}
    then the following statements hold:
    \begin{enumerate}
    \item \label{item:2}
      $\theta$ is a selfadjoint unitary.
    \item \label{item:4}
      $\theta$ implements the \emph{parity automorphism} $\Theta$, i.e. $\theta B(x)
      \theta = -B(x)$ for all $x \in \mathcal{K}^{(n)}$.
    \item \label{item:3}
      For each unitary $R \in \mathcal{B}(\mathcal{K}^{(n)})$ commuting with
      $J$ (i.e. each real orthogonal transformation of
      ($\mathcal{K}^{(n)}$,$J$)) we have
      \begin{equation} \label{eq:64}
        \theta = \det(R) i^n 2^n B(R e_1) \dots B(R e_{2n}).
      \end{equation}
    \end{enumerate}
  \end{prop}

  \begin{proof}
    Item \ref{item:2}. $\theta$ is unitary: Due to $B(h)^* = B(\J h)$ and
    $\J e_a = e_a$ all the operators $B_a$ are selfadjoint, hence
    \begin{equation} \label{eq:35}
      \theta^* = 2^n (-i)^n B(e_{2n}) \cdots B(e_1)
    \end{equation}
    and this implies together with $\{B(e_a), B(e_a)\} = 2 B(e_a)^2 = \Bbb{1}$
    that $\theta \theta^* = \theta^*\theta = \Bbb{1}$. 

    To see that $\theta$ is also selfadjoint, note first that $B(e_a)$
    and $B(e_b)$ anticommute iff $a\neq b$. Now reverse inductively the
    ordering of the factors in Equation (\ref{eq:35}). The first two steps
    lead to 
    \begin{equation}
      \theta = -i^n 2^n B(e_1) B(e_{2d}) \cdots B(e_2) = -i^n 2^n B(e_1) B(e_2)
      B(e_{2n}) \cdots B(e_3) = 
      \dots \ .
    \end{equation}
    Hence we pick up an additional factor $-1$ at each second step (i.e. while
    moving a $B(e_a)$ with $a$ odd to the $a^{\mathrm{th}}$ position). In this
    way we get in total a factor $(-1)^n$, which shows with Equation
    (\ref{eq:35}) that $\theta$ is selfadjoint as stated.  

    Item \ref{item:4}. With a similar argument we can show that for all
    $a=1,\dots,2n$ we have 
    \begin{equation}
      \theta B(e_a) = (-1)^{2n-1} B(e_a) \theta= - B(e_a)
      \theta.  
    \end{equation}
    As a selfadjoint unitary $\theta$ satisfies $\theta^2 =
    \Bbb{1}$. This implies $\theta B(e_a) \theta = - B(e_a)$,
    which completes the proof of item \ref{item:4}.

    Item \ref{item:3}. If $R$ is real orthogonal, the basis $\tilde{e}_a =
    Re_a$ is again a real basis and the reasoning from above applies. Hence
    \begin{equation}
      \theta_R = i^n 2^n B(R e_1) \dots B(R e_{2n})
    \end{equation}
    is a selfadjoint unitary and implements the parity automorphism
    $\Theta$. Only the operators $\pm \theta$ have these properties, and
    therefore we only have to check the sign. To this end consider a quasifree
    state $\rho_S$ with covariance matrix $S$. To complete the proof we have
    to show that $\tr(\rho_S \theta_R) = \det(R) \tr(\rho_S \theta)$.

    According to the properties of the covariance operator $S$, we can
    introduce a real, antisymmetric matrix $\tilde{S}$ by
    \begin{equation}
      i \tilde{S}_{ab} = \langle e_a, (S - \Bbb{1}/2)) e_b \rangle.
    \end{equation}
    With the definition of quasifree states we therefore get
    \begin{align}
      \tr (\rho_S \theta) &= 2^n i^n \tr(\rho_S B(e_1) \cdots B(e_{2n}))  \\
      & = 2^n i^n \sum  \sign( p ) \prod_{j=1}^{n} \langle e_{p(2j-1)} , S e_{p(2j)}
      \rangle, \\
      &= 2^n i^n \sign( p ) \prod_{j=1}^{n} i \tilde{S}_{p(2j-1),p(2j)} \\
      &= 2^n (-1)^n \Pf(\tilde{S})
    \end{align}
    and similarly
    \begin{equation}
      \tr(\rho_S \theta_R) = 2^n (-1)^n \Pf(R^* \tilde{S} R)
    \end{equation}
    where we have identified in abuse of notation $R$ with its matrix
    representation in the basis $e_a$. The statement now follows from the
    properties of Pfaffian (i.e. $\Pf(R^* \tilde{S} R) = \det(R)
    \Pf(\tilde{S})$). 
  \end{proof}

  Note that we have shown in addition that the expectation value of $\theta$
  in a quasifree state $\rho_S$, is given by 
  \begin{equation} \label{eq:62}
    \tr(\rho_S \theta) = 2^n (-1)^n \Pf(\tilde{S}).
  \end{equation}
  This is very useful in the proof of Theorem \ref{thm:3}. The only problem is
  that we have decided about the sign of $\theta$ in an arbitrary way. This
  can be fixed with the following Lemma, which uses the notations introduced
  in Section \ref{sec:bipart-syst-entangl}. 

  \begin{lem} \label{lem:1}
    Consider a bipartite, Fermionic system consisting of $2m$ modes, and a
    quasifree, maximally entangled state $\psi_E$ with covariance matrix $E
    \in \mathcal{B}(\mathcal{K}^{(2m)})$. 
    \begin{enumerate}
    \item 
      There is a  real basis $e_a \in \mathcal{K}^{(2m)}$, $a=1,\dots,4m$ such
      that
      \begin{equation}
        e_1,\dots,e_{2m} \in \mathcal{K}_A^{(m)},\quad e_{2m+1},\dots,e_{4m}
        \in \mathcal{K}_B^{(m)}
      \end{equation}
      and
      \begin{equation} \label{eq:65}
        Y_E e_{2m+j} = e_{2m} \quad \forall j=1,\dots,2m
      \end{equation}
      holds, where $Y_E = -i I_A E I_B$ as in Equation (\ref{eq:55}). 
    \end{enumerate}
  \item 
    If $\theta$ denotes the parity operator given by Equation (\ref{eq:64})
    and this basis, we get
    \begin{equation}
      \langle \psi_E, \theta \psi_E\rangle = (-1)^m.
    \end{equation}
  \end{lem}

  \begin{proof}
    Since $\psi_E$ is maximally entangled by assumption, the off-diagonal
    operator $Y$ is a partial isometry. Hence we can construct the basis $e_a$
    by choosing $e_{2m+1},\dots,e_{4m}$ as an arbitrary real basis of
    $\mathcal{K}_B^{(m)}$ and defining $e_1,\dots,e_{2m}$ by (\ref{eq:65}). In
    this basis $E$ becomes
    \begin{equation}
      E = \frac{1}{2} \left(
      \begin{array}{cc}
        \Bbb{1} & i\Bbb{1} \\
        -i \Bbb{1} & \Bbb{1}.
      \end{array}\right)
    \end{equation}
    Inserting this into (\ref{eq:62}) we get
    \begin{equation} \label{eq:66}
      \langle\psi_E, \theta \psi_E\rangle = 2^{2m} \Pf(\tilde{E})
    \end{equation}
    with 
    \begin{equation} 
      \tilde{E} = \frac{1}{2} \left(
      \begin{array}{cc}
        0 & \Bbb{1} \\
        -\Bbb{1} & 0
      \end{array} \right).
    \end{equation}
    The Pfaffian of $\tilde{E}$ is easy to calculate and gives $(-1)^{m(2m-1)} 2^{-2m}
    = (-1)^m 2^{-2m}$. Inserting this in (\ref{eq:66}) we get the result.
  \end{proof}
  
  \noindent\textit{Proof of Theorem \ref{thm:3}.}  
  Since $P_+ = (\Bbb{1} + \theta)/2$ holds, the probability $p = \tr(\rho_S
  P_+)$ can be calculated as $p = (1 + \tr(\rho_S \theta))/2$ where the sign of
  the parity is chosen by $\langle \psi_E, \theta \psi_E\rangle = 1$. Hence the
  statement follows from Equation (\ref{eq:62}) and Lemma \ref{lem:1}. \hfill
  $\Box$. 

  %%%%%%%%%%%%%%%%%%%%%%%%%%%%%%%%%%%%%%%%%%%%%%%%%%%%%%%%%%%%%%%%%%%%%%%%%%%%%%%%
  \section{The optimality proof}
  \label{sec:optimality-proof}
  %%%%%%%%%%%%%%%%%%%%%%%%%%%%%%%%%%%%%%%%%%%%%%%%%%%%%%%%%%%%%%%%%%%%%%%%%%%%%%%%

  The purpose of this appendix is to provide the proof of the optimality Theorem
  \ref{thm:5}. The crucial assumption is that the diagonal subblocks of the
  covariance operator $S$ satisfy $X = 0$ or $Z = 0$. In both cases the quantity
  $pf$ we have to optimize becomes (this follows directly from Equation
  (\ref{eq:61}), and properties of the Pfaffian and the determinant):
  \begin{equation} \label{eq:67}
    pf(S,m,D,V) = \frac{1}{2^{2m}} \left( \left| \det(D_A Y D_B + V) \right| +
      \left| \det(D_A Y D_B - V) \right| \right).
  \end{equation}
  Here we have identified in abuse of notation the spaces $D_A
  \mathcal{K}^{(2m)}$ and $D_B \mathcal{K}^{(2m)}$ and interpreted the operators 
  $D_A Y D_B \pm V$ as operators on $D_A \mathcal{K}^{(2m)}$, such that the
  determinants have a chance to lead to a non-zero value. The first step in the
  proof is to optimize with keeping $D$ fixed. 

  \begin{prop} \label{prop:1}
    Consider a bipartite Fermionic system consisting of $2m$ modes and a
    quasifree state $\rho_S$ with covariance matrix $S$ satisfying $X_S=0$ or
    $Z_S=0$. Assume in addition that the off-diagonal block has maximal rank $2m$,
    then we have 
    \begin{equation}
      pf(S,m,\Bbb{1},V) \leq pf(S,m,\Bbb{1},\hat{V}) = \prod_{k=1}^{2m}
      \frac{1+\lambda_k}{2} + \prod_{k=1}^{2m}\frac{1-\lambda_k}{2} 
    \end{equation}
    where $\lambda_k$ denotes the eigenvalues of $|Y|$ and $\hat{V}$ is the
    partial isometry defined in Equation (\ref{eq:68}).
  \end{prop}

  \begin{proof}
    Let us choose a real basis $e_a$ in $\mathcal{K}^{(2m)}$ such that $Y$ is
    diagonalized, i.e. $\langle e_j, Y e_{k+2m}\rangle = \delta_{jk}
    \lambda_k$. The matrix elements $V_{jk} = \langle e_j, V e_{k+2m}\rangle$ of
    $V$ in the same basis form a orthogonal matrix. Hence the determinants in
    Equation (\ref{eq:67}) become
    \begin{equation}
      \det(Y \pm V) = \sum_p \epsilon_p \prod_{j=1}^{2m} (V_{jp(j)} \pm
      \delta_{jp(j)} \lambda_j),
    \end{equation}
    where the sum is taken over all permutations of $1,\dots,2m$. This can be
    rewritten as a polynomial in the $\lambda_j$: 
    \begin{equation}
      \det(Y + V) = \sum_{\Sigma \subset \{1,\dots,2m\}} \prod_{j \in \Sigma}
      \lambda_j \det(V_\Sigma). 
    \end{equation}
    $V_\Sigma$ denotes the matrix we get from $V$ if we remove the
    $j^{\mathrm{th}}$ row and column for all $j \in \Sigma$. Now assume that
    both determinants have the same sign. Then Equation (\ref{eq:67}) leads to
    \begin{equation} \label{eq:69}
      pf(S,m,\Bbb{1},V) = 2 \sum_{\substack{\Sigma \subset \{1,\dots,2m\}\\
          \Sigma\ \text{even}}} \prod_{j \in \Sigma} \lambda_j \det(V_\Sigma).
    \end{equation}
    Now note that $0 < \lambda_j \leq 1$ for all $j$ and $V_\Sigma$ is a
    submatrix of an orthogonal matrix. Hence $|\det V_\Sigma| \leq 1$ with
    equality iff $V_\Sigma$ is orthogonal again. The expression in (\ref{eq:69})
    is maximal iff $\det V_\Sigma = 1$ for all $\Sigma$. But this is implies
    that all submatrix $V_\Sigma$ are orthogonal. This is only possible if $V$ is
    diagonal. A similar reasoning holds if both determinants in (\ref{eq:67})
    have opposite signs (instead of (\ref{eq:69}) we get a sum over subsets
    with an odd cardinality). 

    This shows that $Y$ has to be diagonal with eigenvalues $\pm 1$. The optimal
    value of $pf(S,m,\Bbb{1},V)$ then becomes
    \begin{equation}
      \mathrm{fid}(\Gamma_+,\Gamma_-) + \mathrm{fid}(\Gamma_-,\Gamma_+)
    \end{equation}
    with
    \begin{equation}
      \mathrm{fid}(\Gamma_+,\Gamma_-) = \left| \prod_{j \in \Gamma_+} (1 +
        \lambda_j) \prod_{k \in  \Gamma_-} (-1 + \lambda_k) \right| 
      = \prod_{j \in \Gamma_+} (1 +\lambda_j) \prod_{k \in \Gamma_-}
      (1-\lambda_k) \label{eq:70} 
    \end{equation}
    where $\Gamma_+$ and $\Gamma_-$ are disjoint sets of integers satisfying
    $\Gamma_+ \cup \Gamma_- = \{1,\dots,2m\}$. Also note that the second
    equation in (\ref{eq:70}) holds, due to $0 < \lambda_j \leq 1$. The
    proposition is proved if we can show that
    \begin{equation}
      \mathrm{fid}(\{1,\dots,2m\}, \emptyset) +
      \mathrm{fid}(\emptyset,\{1,\dots,2m\}) \geq
      \mathrm{fid}(\Gamma_+,\Gamma_-) + \mathrm{fid}(\Gamma_-,\Gamma_+)
    \end{equation}
    holds. To prove this inequality rewrite $\mathrm{fid}$ as
    \begin{align}
      \mathrm{fid}(\Gamma_+,\Gamma_-) &= \mathrm{fid}(\Gamma_+, \emptyset)
      \sum_{\gamma \subset \Gamma_-} (-1)^{|\gamma|} \prod_{j\in\gamma}
      \lambda_j \\
      &= \mathrm{fid}(\emptyset,\Gamma_-) \sum_{\gamma \subset \Gamma_+}
      \prod_{j\in\gamma} \lambda_j.  
    \end{align}
    This leads to 
    \begin{equation} \label{eq:71}
      \mathrm{fid}(\Gamma_+,\Gamma_-) + \mathrm{fid}(\Gamma_-,\Gamma_+) =
      \sum_{\gamma \subset \Gamma_-} \bigl((-1)^{|\gamma|}
      \mathrm{fid}(\Gamma_+,\emptyset) +
      \mathrm{fid}(\emptyset,\Gamma_+)\bigr)\prod_{j\in\gamma} \lambda_j. 
    \end{equation}
    The definition of $\mathrm{fid}$ shows that
    \begin{equation}
      \mathrm{fid}(\Gamma_+,\emptyset) \geq \mathrm{fid}(\emptyset,\Gamma_+)
    \end{equation}
    holds. Hence each summand in (\ref{eq:71}) belonging to $\gamma$ with odd
    cardinality is negative. In contrast to that we get
    \begin{equation}
      \mathrm{fid}(\{1,\dots,2m\}, \emptyset) +
      \mathrm{fid}(\emptyset,\{1,\dots,2m\}) = \sum_{\gamma \subset \Gamma_-}
      \bigl(\mathrm{fid}(\Gamma_+,\emptyset) + (-1)^{|\gamma|}
      \mathrm{fid}(\emptyset,\Gamma_+)\bigr) \prod_{j\in\gamma} \lambda_j
    \end{equation}
    where all summands are positive. This completes the proof.
  \end{proof}
  
  \noindent\textit{Proof of Theorem \ref{thm:5}.} We are now ready to complete
  the proof of Theorem \ref{thm:5}. Hence let us consider $2L$ modes in the
  quasifree $\rho_S$, an integer $m < L$ and a real projection $D = D_A \oplus
  D_B \in \mathcal{K}^{(2m)}$. If $V_{DSD}$ denotes the partial isometry given
  by the polar decomposition of $D_AYD_B$, we get with Proposition \ref{prop:1}:
  \begin{equation}
    pf(S,m,D,V_{DSD}) = \frac{1}{2^{2m}} \prod_{k=1}^{2m} \frac{1+\tilde{\lambda}_k}{2} +
    \prod_{k=1}^{2m}\frac{1-\tilde{\lambda}_k}{2}
    = \frac{1}{2^{2m}} \sum_{\substack{\Sigma \subset \{1,\dots,2m\}\\ |\Sigma|\ \text{even}}}
    \prod_{j \in \Sigma} \tilde{\lambda}_j
  \end{equation}
  where $\tilde{\lambda}_j$ are the eigenvalues of $|Y_{DSD}|$ given in
  decreasing order. Obviously this quantity is optimized, if the
  $\tilde{\lambda}_j$ are as big as possible. Hence the theorem is proved if we
  can show that each $D$ satisfies
  \begin{equation}
    \lambda_j \geq \tilde{\lambda}_j\quad \forall j=1,\dots,2m
  \end{equation}
  with equality only if $D = \hat{D}$. Here $\lambda_1, \dots, \lambda_{2L}$ are
  the eigenvalues of $|Y|$. 

   To get this bound note first that instead of $|Y|$ and $|Y_{DSD}|$ we can
   look at the eigenvalues of $Y^*Y$ and $Y_{DSD}^*Y_{DSD}$. To get the
   corresponding bound we will use the  Courant-Fischer theorem (Theorem 4.2.11
   of \cite{MR832183}) which states that the $k^{\mathrm{th}}$ highest
   eigenvalue $\lambda_k$ of a hermitian matrix $A$ is given by
  \begin{equation}
    \lambda_k = \sup_F \inf_{\substack{x \neq 0\\ Fx=0}} \frac{\langle x, A
        x\rangle}{\langle x, x \rangle}  
  \end{equation}
  where the supremum is taken over all rank $k-1$ projections $F$. Now we get
  \begin{align}
    \frac{\langle x, Y^*Y x\rangle}{\langle x, x \rangle} &\geq \frac{\langle D_Bx, Y^*Y
      D_Bx\rangle}{\langle x, x \rangle}  \\
    &\geq \frac{\langle x, D_B Y^*D_A Y D_B x\rangle}{\langle x, x \rangle} + \frac{\langle x, D_B
      Y^*(\Bbb{1} - D_A) Y D_B x\rangle}{\langle x, x \rangle} \\
    &\geq \frac{\langle x, D_B Y^*D_A Y D_B x\rangle}{\langle x, x \rangle},
  \end{align}
  where we have used the fact that $\Bbb{1}-D_A$ is a positive operator. This leads to
  \begin{equation}
    \sup_F \inf_{\substack{x \neq 0\\ Fx=0}} \frac{\langle x, Y^*Y x\rangle}{\langle x, x \rangle}
    \geq \sup_F \inf_{\substack{x \neq 0\\ Fx=0}} \frac{\langle x, D_B Y^* D_A Y D_B
      x\rangle}{\langle x, x \rangle}.  
  \end{equation}
  Hence the $k^{\mathrm{th}}$ highest eigenvalue of $Y^*Y$ is greater or equal to the
  $k^{\mathrm{th}}$ highest eigenvalue of $D_B Y^* D_A Y D_B$. Since this bound
  is attained by $\hat{D} = \hat{D}_A \oplus \hat{D}_B$  the theorem is
  proved.\hfill $\Box$

  \medskip
  \noindent\textit{Proof of Corollary \ref{kor:1}.} We show that increasing $m$
  by two decreases $pf$, i.e.
  \begin{equation}
    pf(S,m,\hat{D},\hat{V}) \geq pf(S,m+2,\hat{D},\hat{V}).
  \end{equation}
  To this end let us introduce first the polynomial
  \begin{equation} 
    g = \frac{1}{2^{2m}} \sum_{\substack{\Sigma \subset \{1,\dots,2m\}\\
        |\Sigma|\ \text{odd}}}  \prod_{j \in \Sigma} \lambda_j =
    \prod_{j=1}^{2m} \frac{1+\lambda_j}{2} -  \prod_{j=1}^{2m}
    \frac{1-\lambda_j}{2}  \leq pf(S,m,\hat{D},\hat{V}) \label{eq:73}
  \end{equation}
  Then we can express $pf(S,m+2,\hat{D},\hat{V})$ as
  \begin{equation}
    pf(S,m+2,\hat{D},\hat{V}) = \frac{1}{4} ( pf(S,m,\hat{D},\hat{V}) +
    pf(S,m,\hat{D},\hat{V}) \lambda_{m+1}\lambda_{m+2} + g \lambda_{m+1} 
    +g \lambda_{m+2}).
  \end{equation}
  The statement follows from Equation (\ref{eq:73}) and the fact that
  $\lambda_{m+1} \leq 1$ and $\lambda_{m+2} \leq 1$.\hfill $\Box$

\end{appendix}

\bibliographystyle{mk} \bibliography{qinf}
\end{document}